\documentclass[11pt]{article}

\usepackage{amsfonts,latexsym,amsthm,amssymb,amsmath,amscd,euscript,tikz,mathtools}
\usepackage{framed}
\usepackage[margin=1in]{geometry}
\usepackage{color}
\usepackage[colorlinks=true,citecolor=blue,linkcolor=blue]{hyperref}
\usepackage{enumitem}
\usepackage{siunitx}
\usepackage{textcomp}
\usepackage{physics}
\usepackage{mathrsfs}
\usepackage{dsfont}
\usepackage[ruled,vlined,linesnumbered]{algorithm2e}
\usepackage{titlesec}
\usepackage{tikz-cd}
\usepackage{thmtools}
\usepackage{thm-restate}
\usepackage{cleveref}
\usepackage{stmaryrd}
\usepackage{graphicx}
\usetikzlibrary{positioning,chains,fit,shapes,calc}

\allowdisplaybreaks[1]

\newtheorem{theorem}{Theorem}[section]
\newtheorem{proposition}[theorem]{Proposition}
\newtheorem{lemma}[theorem]{Lemma}

\newtheorem{corollary}[theorem]{Corollary}

\newtheorem{definition}[theorem]{Definition}
\newtheorem{example}[theorem]{Example}
\newtheorem{remark}[theorem]{Remark}

\theoremstyle{remark}

\newcommand{\cA}{\mathcal{A}}\newcommand{\cB}{\mathcal{B}}
\newcommand{\cC}{\mathcal{C}}

\newcommand{\cL}{\mathcal{L}}

\newcommand{\bF}{\mathbb{F}}

\newcommand{\bN}{\mathbb{N}}

\newcommand{\bZ}{\mathbb{Z}}

\newcommand{\1}{\mathds{1}}

\newcommand{\poly}{\operatorname{poly}}

\newcommand{\dis}{\operatorname{dis}}

\newcommand{\nc}{\newcommand}

\nc{\on}{\operatorname}
\nc{\Spec}{\on{Spec}}
\nc{\Aut}{\textit{Aut}}
\nc{\id}{\textit{id}}
\nc{\chr}{\on{char}}
\nc{\im}{\on{im}}
\nc{\Hom}{\on{Hom}}
\nc{\lcm}{\on{lcm}}
\nc{\dual}[1]{\prescript{t}{}{#1}}
\nc{\transpose}[1]{{#1}^{\intercal}}
\nc{\Sym}{\on{Sym}}
\nc{\End}{\on{End}}
\nc{\stab}{\on{stab}}
\nc{\Li}{\on{Li}}
\nc{\spn}{\on{span}}
\nc{\sgn}{\on{sgn}}
\nc{\supp}{\on{supp}}
\nc{\Unif}{\on{Unif}}

\title{
  Quantum LDPC Codes of Almost Linear Distance \\ via Homological Products\thanks{Research supported in part by a ONR grant N00014-24-1-2491 and a UC Noyce initiative award. L.~Golowich is supported by a National Science Foundation Graduate Research Fellowship under Grant No.~DGE 2146752. This work was done in part while the authors were attending the Spring 2024 programs at the Simons Institute for the Theory of Computing.}
}
\author{Louis Golowich \\
  UC Berkeley \\
  \href{mailto:lgolowich@berkeley.edu}{\texttt{lgolowich@berkeley.edu}}
  \and
  Venkatesan Guruswami \\
  UC Berkeley \\
  \href{mailto:venkatg@berkeley.edu}{\texttt{venkatg@berkeley.edu}}
}

\begin{document}

\pagenumbering{gobble}

\maketitle
\thispagestyle{empty}

\begin{abstract}
  We present new constructions of quantum codes of linear or close-to-linear distance and dimension with low-weight stabilizers. Only a few constructions of such codes were previously known, and were primarily based on a specific operation from homological algebra, namely the balanced product. In contrast, our constructions are based on a more basic and widely used product, namely the homological product (i.e.~the tensor product of chain complexes). Our results help address the natural question: When do homological products preserve good code distance?

  Our first main result constructs asymptotically good $[[N,\Theta(N),\Theta(N)]]$ quantum codes with small polynomial stabilizer weight from homological products of codes with a property called \textit{product-expansion}. This notion was recently introduced and used to bound the distance of balanced product quantum codes; we apply it instead to homological products.

  For every $\epsilon>0$, our second main result constructs close-to-linear distance $[[N,N^{1-\epsilon},N^{1-\epsilon}]]$ (subsystem) quantum LDPC codes with constant stabilizer weight from iterated homological products of a constant-sized quantum locally testable code. The key insight here is that by using subsystem codes (but still with constant-weight stabilizers), we can circumvent a particular obstruction that limited the distance of many prior product code constructions to at most~$\tilde{O}(\sqrt{N})$.
\end{abstract}

\newpage

\tableofcontents

\newpage

\pagenumbering{arabic}


\section{Introduction}
Quantum low-density parity-check (qLDPC) codes provide one of the most promising means for achieving efficient quantum fault-tolerance. Such codes are defined to support error detection and correction via sparse queries (i.e.~stabilizer measurements) to code states. Each such query involves, and can therefore only propagate errors across, a small number of code components. However, this qLDPC condition has proven difficult to achieve, and asymptotically optimal qLDPC codes were only recently constructed following a line of breakthrough works \cite{hastings_fiber_2021,panteleev_quantum_2022,breuckmann_balanced_2021,panteleev_asymptotically_2022,leverrier_quantum_2022-1,dinur_good_2023}. The latter three of these works use nearly identical techniques, and still provide the only known asymptotically good qLDPC codes, meaning that the code dimension and distance scale linearly in the block length, and the stabilizers have constant weight. It remains an open question to find new constructions of qLDPC codes with good parameters. Such alternative constructions may yield better practical parameters or have properties useful for fault-tolerant computation, and would also be of independent mathematical interest.

In this paper, we develop new constructions of qLDPC codes with linear or close-to-linear dimension and distance based on \textit{homological products} (i.e.~tensor products of chain complexes), a well-known construction from homological algebra that generalizes classical tensor codes. Our constructions are based on new characterizations that we develop for conditions under which homological products preserve (close-to) linear quantum code distance.

Homological products have long been used to construct qLDPC codes, notably including hypergraph product codes \cite{tillich_quantum_2014}, which are simply homological products of 2-term chain complexes associated to classical LDPC codes. Hypergraph products of good classical LDPC codes give length-$N$ quantum codes of dimension $\Theta(N)$ and distance $\Theta(\sqrt{N})$ \cite{tillich_quantum_2014}; the well-known surface and Toric codes \cite{kitaev_fault-tolerant_2003} can also be viewed as hypergraph products of repetition codes.

However, it was a longstanding open question to construct qLDPC codes with distance greater than $\Theta(\sqrt{N})$. This question was resolved by the line of work described above leading to asymptotically good qLDPC codes, using \textit{balanced} or \textit{lifted products}, which can be viewed as quotients of hypergraph products by a group symmetry. Such balanced products remain the only known technique for obtaining qLDPC codes with (even close to) linear dimension and distance. Furthermore, balanced products are a fairly rigid construction, in that they require classical codes respecting a large group symmetry. The optimal constructions \cite{panteleev_asymptotically_2022,leverrier_quantum_2022-1,dinur_good_2023} also require classical codes with a certain strong property called \textit{product-expansion} \cite{kalachev_two-sided_2023}, which is difficult to obtain in practice.

In light of these limitations of balanced products, we study the question of whether homological products can be used to construct qLDPC codes with comparably good parameters. Homological products are more flexible in that they require no group symmetry. At a more foundational level, it is a mathematically interesting question to determine when homological product codes have large distance. Indeed, as quantum codes can be equivalently viewed as chain complexes or as topological spaces (see also \cite{freedman_building_2021}), this question can be phrased as asking how the systolic and cosystolic distance of a topological space (see Definition~\ref{def:sysdisexp}) behave under topological products.

In addition to its fundamental nature, this question is also of practical interest. For instance, the classical analogue of homological product codes, namely tensor codes, are widely studied and used throughout classical coding theory. Better understanding quantum homological products may yield similar benefits for quantum fault tolerance and beyond.

As mentioned above, many qLDPC code constructions based on homological products had distance at most $O(\sqrt{N})$. However, a few works have surpassed this barrier. Specifically, \cite{bravyi_homological_2014} showed that a length-$N$ homological product of two quantum random codes has distance $\Theta(N)$ and stabilizer weight $\Theta(\sqrt{N})$. \cite{kaufman_new_2021} used homological products of chain complexes from certain high-dimensional expanders \cite{lubotzky_ramanujan_2005,lubotzky_explicit_2005} to obtain qLDPC codes with distance $\Theta(N\cdot\poly\log N)$ and constant stabilizer weight. \cite{delfosse_union-find_2021} also showed how to obtain homological product codes with stabilizer weight $w$ and distance $\Theta(\sqrt{N\cdot w})$.

In this paper, we present two main results, each of which uses a different technique to construct homological product codes of linear or close-to-linear distance with low-weight stabilizers. The following two sections summarize these respective results. A more detailed overview, with necessary background notions and notation, is provided in Section~\ref{sec:techoverview}.

\subsection{Results on Product-Expanding Codes}
Our first main result is a simplification and generalization of the work of \cite{bravyi_homological_2014}, from which we obtain codes of linear dimension and distance, with stabilizer weight growing as a small polynomial in the block length. Interestingly, we bound the distance of these codes using bounds of \cite{polishchuk_nearly-linear_1994,kalachev_two-sided_2023,kalachev_personal_2024} on the product-expansion property mentioned earlier. As a result, we are able to derandomize the \cite{bravyi_homological_2014} result using Reed-Solomon codes, extend it to bound local testability (see Definition~\ref{def:CSScode}) in addition to distance, and also generalize it to higher-dimensional products, which have lower stabilizer weight. At a high level, these results suggest that product-expansion, viewed as a pseudorandom property, captures and generalizes the behavior of random homological products observed by \cite{bravyi_homological_2014}. The informal theorem statement below summarizes the code constructions described above:

\begin{theorem}[Informal statement of Theorem~\ref{thm:peprodinf}]
  \label{thm:peprodintro}
  Let $Q^1,\dots,Q^t$ consist of either:
  \begin{enumerate}
  \item\label{it:pe2} $t=2$ uniformly random quantum CSS codes over an arbitrary alphabet $\bF_q$, or
  \item\label{it:peRS} $t=2$ quantum Reed-Solomon CSS codes, or
  \item\label{it:pehigh} An arbitrary constant number $t\in\bN$ of random quantum CSS codes over a sufficiently large alphabet $\bF_q$.
  \end{enumerate}
  Then the homological product of $Q^1,\dots,Q^t$ gives an asymptotically good\footnote{Recall that an $[[N,K,D]]_q$ code is a quantum code on qudits of local dimension $q$ with block length $N$, dimension (i.e.~number of logical qudits) $K$, and distance $D$. The code is said to be asymptotically good if $K,D=\Theta(N)$.} $[[N,\Theta(N),\Theta(N)]]_q$ CSS code with stabilizer weight $O(N^{1/t})$. Furthermore, if $t=2$, then this product code is locally testable with $O(\sqrt{N})$-weight queries and constant soundness.
\end{theorem}

In Theorem~\ref{thm:peprodintro}, the proof of item~\ref{it:pe2} applies the product-expansion result of \cite{kalachev_two-sided_2023}. The proof of item~\ref{it:peRS} applies the product-expansion result of \cite{polishchuk_nearly-linear_1994}, which is perhaps surprising given that \cite{kalachev_two-sided_2023} observe that this result is too weak to apply to known distance bounds for balanced product constructions \cite{panteleev_asymptotically_2022,leverrier_quantum_2022-1,dinur_good_2023}. The proof of item~\ref{it:pehigh} applies a forthcoming product-expansion result of Kalachev \& Panteleev \cite{kalachev_personal_2024}.

\subsection{Results on Subsystem QLDPC Codes}
Our second main result provides a new method of taking homological products that preserve almost-linear distance by using \textit{subsystem codes}, which only encode messages into a subspace of the entire logical code space. Our analysis of this method builds on techniques of \cite{kaufman_new_2021}. As a result, we obtain a new construction of qLDPC codes of close-to-linear dimension and distance with constant stabilizer weight using an iterative construction, which is based on iterative homological products of a constant-sized code. We provide an informal statement below:

\begin{theorem}[Informal statement of Theorem~\ref{thm:iterativeinf}]
  \label{thm:iterativeintro}
  For every $\epsilon>0$, there exists a constant-sized code $\cC$ and an infinite sequence of qLDPC subsystem codes $(Q_i)_{i\in\bN}$ with parameters
  \begin{equation*}
    [[N_i,N_i^{1-\epsilon},N_i^{1-\epsilon}]]_2
  \end{equation*}
  and with constant stabilizer weight $O(1)$ (independent of $\epsilon$), such that each $Q_i$ is obtained by applying the stabilizer-weight-reduction transformation of \cite{hastings_quantum_2023} to the homological product of $Q_{i-1}$ with $\cC$.
\end{theorem}

The weight-reduction step in Theorem~\ref{thm:iterativeintro} is needed to keep the stabilizer weight constant in each iteration, as in general the stabilizer weight of a homological product code grows as the sum of the stabilizer weights of the input codes.

The key idea behind the construction in Theorem~\ref{thm:iterativeintro} is that by using subsystem codes, we are able to circumvent the $\tilde{\Theta}(\sqrt{N})$ distance barrier that limited most prior homological product code constructions. Specifically, we observe that this barrier only applies to certain logical operators (i.e.~codewords) within the code space. Therefore, we show that appropriate homological product codes will still have good distance for a large subspace of the logical operators. Thus appropriate subsystem codes of these products will still have good distance. We prove our specific distance bounds by adapting the techniques of \cite{kaufman_new_2021}, which showed a similar result, but that did not use subsystem codes. 

To prove Theorem~\ref{thm:iterativeintro}, we set $\cC$ to be a constant-sized quantum locally testable code (qLTC; see Definition~\ref{def:CSScode}) such as one given by \cite{dinur_expansion_2024} (which in turn applies a forthcoming result of \cite{kalachev_personal_2024}); such a qLTC (which we emphasize is just constant-sized) is needed for our subsystem distance bound described above. We apply this distance bound to show that the distance of each $Q_i$ remains close-to-linear (i.e.~$\geq N_i^{1-\epsilon}$) in the block length $N_i$.

This iterative construction is reminiscent of a line of work in classical pseudorandomness, coding theory, and complexity theory, which iteratively builds larger objects with properties of interest from smaller or weaker ones. For instace, \cite{ben-sasson_robust_2004} constructed classical locally testable codes by iterative tensoring of a small classical code; at a high level, Theorem~\ref{thm:iterativeintro} uses a related approach to construct qLDPC codes. Other notable iterative constructions in the literature include explicit spectral expanders from the zig-zag product \cite{reingold_entropy_2002}, Dinur's proof of the PCP theorem \cite{dinur_pcp_2007}, and undirected connectivity in logarithmic space \cite{reingold_undirected_2008}. Theorem~\ref{thm:iterativeintro} can therefore be viewed as an analogue of these results for qLDPC codes. However, our construction does need to start with a strong constant-sized object, namely, a constant-sized qLTC with sufficiently good parameters. It is an interesting open question whether there exists an iterative construction with a weaker starting object.

\section{Technical Overview}
\label{sec:techoverview}
In this section, we provide a more detailed technical overview of our results. After presenting some necessary preliminaries in Section~\ref{sec:qcodeprelim}, we describe the results from our two main techniques in Section~\ref{sec:peprodinf} and Section~\ref{sec:subprodinf} respectively.

\subsection{Quantum Code Preliminaries}
\label{sec:qcodeprelim}
In this section, we briefly summarize some necessary background on quantum codes; more details can be found in Section~\ref{sec:prelim}.

A length-$n$ quantum CSS code consists of a pair $Q=(Q_X,Q_Z)$ of vector spaces $Q_X,Q_Z\subseteq\bF_q^n$ with $Q_X^\perp\subseteq Q_Z$, where $Q_X^\perp:=\{c\in\bF_q^n:c\cdot x=0\;\forall x\in Q_X\}$. The code $Q$ has dimension $k=\dim(Q_Z)-\dim(Q_X^\perp)$ and distance $d=\min_{c\in(Q_X\setminus Q_Z^\perp)\cup(Q_Z\setminus Q_X^\perp)}|c|$. We summarize these parameters by saying that $Q$ is an $[[n,k,d]]_q$ code.

If $H_X\in\bF_q^{m_X\times n},\;H_Z\in\bF_q^{m_Z\times n}$ are parity-check matrices for $Q_X,Q_Z$ respectively, meaning that $Q_X=\ker H_X$ and $Q_Z=\ker Q_Z$, then the locality (i.e.~stabilizer weight) $w$ of $Q$ equals the maximum weight of any row of $H_X$ or $H_Z$. We say $Q$ is locally testable with locality $w$ and soundness $\rho>0$ if for $\alpha\in\{X,Z\}$, it holds for every $e\in\bF_q^n$ that $|H_\alpha e|/m_\alpha\geq\rho\cdot\dis(e,Q_\alpha)/n$, where $\dis(e,Q_\alpha)=\min_{c\in Q_\alpha}|e-c|$. Thus local testability implies that the weight of an error on a code state can be estimated by simply measuring a random stabilizer.

\subsection{Homological Products of Product-Expanding Codes}
\label{sec:peprodinf}
In this section, we describe our results on homological products of chain complexes associated to product-expanding classical codes. For these results, we follow \cite{bravyi_homological_2014} in working with homological products of single-sector chain complexes, defined below.

\begin{definition}
  A \textbf{single-sector chain complex $\cC_*=(C,\partial^{\cC})$ over $\bF_q$} consists of a $\bF_q$-vector space $C$ and a linear \textbf{boundary map $\partial^{\cC}:C\rightarrow C$} satisfying $(\partial^{\cC})^2=0$. The associated \textbf{cochain complex} is $\cC^*=(C,\delta^{\cC}=(\partial^{\cC})^\top)$.

  For single-sector complexes $\cA,\cB$ over a field $\bF_q$ of characteristic $2$, the \textbf{homological product $\cC=\cA\otimes\cB$} is the single-sector complex given by $C=A\otimes B$ and $\partial^{\cC}=\partial^{\cA}\otimes I+I\otimes\partial^{\cB}$.
\end{definition}

A single-sector chain complex $\cC$ has a naturally associated quantum CSS code given by $Q_X=\ker\delta^{\cC}$ and $Q_Z=\ker\partial^{\cC}$. The locality (i.e.~stabilizer weight) of this code equals the maximum Hamming weight of any row or column of $\delta^{\cC}$. Conversely, given a CSS code $Q=(Q_X,Q_Z)$ with $\dim(Q_X)=\dim(Q_Z)$, we can construct a single-sector chain complex $\cC$ with associated code $Q$ by letting $\partial^{\cC}=H_X^\top H_Z$, where $H_X,H_Z$ are full-rank parity-check matrices for $Q_X,Q_Z$ respectively, meaning that $Q_X=\ker H_X$ and $Q_Z=\ker Q_Z$.

Note that by definition, if $\cA$ and $\cB$ have respective localities $w^{\cA}$ and $w^{\cB}$, then $\cC$ has locality $w^{\cC}\leq w^{\cA}+w^{\cB}$. In particular, if $\dim(A)=\dim(B)=n$, then $\dim(C)=n^2$ and $w^{\cC}\leq 2n$.

Following this observation, \cite{bravyi_homological_2014} showed that with high probability, the homological product of single-sector complexes associated to two random CSS codes gives a $[[N,\Theta(N),\Theta(N)]]$ CSS code with locality (i.e.~stabilizer weight) $\Theta(\sqrt{N})$. We strengthen and generalize this result in multiple ways: we show that these product codes are in fact locally testable, we show that the same result holds with Reed-Solomon codes (which are explicit) in place of the random codes, and we extend the result to products of more than two complexes\footnote{This result on higher products applies a forthcoming result of \cite{kalachev_personal_2024} (see Theorem~\ref{thm:petrand})}. Formally, our results are summarized in Theorem~\ref{thm:peprodinf} below. Here, given $k\in\bN$ and a prime power $q$, a length-$q$ quantum Reed-Solomon code is a CSS code with $Q_X=Q_Z=\textrm{RS}(q,k)$, where
\begin{equation*}
  \textrm{RS}(q,k) = \{(f(x))_{x\in\bF_q}:f(X)\in\bF_q[X],\deg(f)<k\}.
\end{equation*}


\begin{theorem}[Informal statement of Corollary~\ref{cor:ssperandom} Corollary~\ref{cor:sspeRS}, Corollary~\ref{cor:sspemanyrandom}]
  \label{thm:peprodinf}
  Let $\cC_1,\cC_2$ be single-sector chain complexes associated to respective CSS codes $Q^1,Q^2$ from either of the two constructions below
  \begin{enumerate}
  \item\label{it:randinf} $Q^1,Q^2$ are uniformly random $[[n,\Theta(n),\Theta(n)]]$ CSS codes, or
  \item $Q^1,Q^2$ are $[[n,k_i=\Theta(n),\Theta(n)]]$ quantum Reed-Solomon codes for $i=1,2$ with $|k^1-k^2|=\Theta(n)$.
  \end{enumerate}
  Then the homological product $\cC_1\otimes\cC_2$ has (with high probability in case~\ref{it:randinf} above) an associated $[[N=n^2,\Theta(N),\Theta(N)]]$ CSS code that is locally testable with locality $\Theta(\sqrt{N})$ and soundness $\Omega(1)$.

  Furthermore, for an arbitrary constant $t\in\bN$, if $Q^1,\dots,Q^t$ are uniformly random $[[n,\Theta(n),\Theta(n)]]_q$ CSS codes over a sufficiently large alphabet $\bF_q$ with associated single-sector complexes $\cC_1,\dots,\cC_t$, then the homological product $\cC_1\otimes\cdots\otimes\cC_t$ gives a $[[N=n^t,\Theta(N),\Theta(N)]]_q$ code with stabilizer weight $\Theta(N^{1/t})$.
\end{theorem}

Whereas \cite{bravyi_homological_2014} proved their distance bound with a specialized analysis for the product of two random CSS codes, we prove the results in Theorem~\ref{thm:peprodinf} in a unified manner, by showing that the distance and testability of homological product codes can in general be bounded using the \textit{product-expansion} of the underlying classical codes, defined below. Therefore our proof of Theorem~\ref{thm:peprodinf} is relatively general and concise; the major difficulty lies in proving product-expansion, for which we apply results of \cite{polishchuk_nearly-linear_1994,kalachev_two-sided_2023,kalachev_personal_2024}.

To define product-expansion, we need the following notation: for classical codes $(C_i\subseteq\bF_q^{n_i})_{i\in[t]}$, let
\begin{equation*}
  C^{(i)}=\bF_q^{n_1}\otimes\cdots\otimes\bF_q^{n_{i-1}}\otimes C_i\otimes\bF_q^{n_{i+1}}\otimes\cdots\otimes\bF_q^{n_t}.
\end{equation*}
Also, for $c\in C^{(i)}$, let $|c|_i\in[0,\prod_{j\neq i}n_j]$ denote the number of nonzero direction-$i$ columns in $c$.

\begin{definition}[\cite{kalachev_two-sided_2023}]
  \label{def:peinf}
  The \textbf{product-expansion $\rho$} of a collection $(C_i\subseteq\bF_q^{n_i})_{i\in[t]}$ of classical codes is the largest real number $\rho\geq 0$ such that for every $c\in C^{(1)}+\cdots+C^{(t)}$, there exists a decomposition
  \begin{equation*}
    c=c_1+\cdots+c_t
  \end{equation*}
  with each $c_i\in C^{(i)}$ such that
  \begin{equation*}
    |c| \geq \rho\sum_{i\in[t]}n_i|c_i|_i.
  \end{equation*}
\end{definition}

When $t=1$, the product-expansion of a single code simply equals its relative distance. Hence product-expansion is a sort of high-dimensional generalization of code distance. The product-expansion of two random classical codes (see \cite{kalachev_two-sided_2023}) was used by \cite{panteleev_asymptotically_2022}, and subsequently in a similar manner by \cite{leverrier_quantum_2022-1,dinur_good_2023}, to prove linear distance for quantum codes obtained from balanced products of classical LDPC codes \cite{breuckmann_balanced_2021}. However, whereas these results apply product-expansion to constant-sized ``local'' codes within the larger construction, we apply product-expansion directly on growing codes.

Our main result bounding homological product code distance via product-expansion is stated below.

\begin{theorem}[Informal statement of Theorem~\ref{thm:sspe} with Proposition~\ref{prop:sskunneth}]
  \label{thm:sspeinf}
  Fix a constant $t\in\bN$, and for $i\in[t]$ let $\cC_i$ be a single-sector chain complexe with associated $[[n_i,k_i,d_i]]$ CSS code $Q^i$. For every $i\in[t]$ assume that the collections $({Q^1_X}^\perp,\dots,{Q^{i-1}_X}^\perp,Q^i_Z)$  and $({Q^1_Z}^\perp,\dots,{Q^{i-1}_Z}^\perp,Q^i_X)$ have product-expansion at least some $\rho>0$. Then the CSS code $Q$ associated to the homological product $\cA=\cC_1\otimes\cdots\otimes\cC_t$ has parameters
  \begin{equation*}
    \left[\left[N=\prod_{i\in[t]}n_i,\; K=\prod_{i\in[t]}k_i,\; D\geq\rho^t\cdot\prod_{i\in[t]}n_i\right]\right].
  \end{equation*}

  Furthermore, if $t=2$, then if $({Q^1_Z}^\perp,{Q^2_Z}^\perp)$ and $({Q^1_X}^\perp,{Q^2_X}^\perp)$ both have product-expansion at least $\rho'>0$, and if the classical codes $Q^1_Z,Q^2_Z,Q^1_X,Q^2_X$ all have relative distance at least $\Delta\in[0,1]$, then the product code $Q$ is locally testable with locality $\leq\max\{n_1,n_2\}$ and soundness $\geq\rho'\cdot\min\{\rho',\Delta\}$.
\end{theorem}

Theorem~\ref{thm:peprodinf} follows immediately from Theorem~\ref{thm:sspeinf} with the product-expansion results of \cite{polishchuk_nearly-linear_1994,kalachev_two-sided_2023,kalachev_personal_2024}. Specifically, \cite{polishchuk_nearly-linear_1994} showed that a pair of Reed-Solomon codes for which the sum of the dimensions is less than the block length $n$ have good (i.e.~constant) product-expansion. \cite{kalachev_two-sided_2023} showed that a pair of random classical codes over arbitrary finite fields have good product-expansion. In a forthcoming work \cite{kalachev_personal_2024}, Kalachev \& Panteleev show that for every fixed $t\in\bN$, a set of $t$ random classical codes over sufficiently large finite fields have good product-expansion. Each of these three results yields a claim in Theorem~\ref{thm:peprodinf}.

It is perhaps interesting that we are able to prove linear code distance using the bound of \cite{polishchuk_nearly-linear_1994} on the product-expansion of Reed-Solomon codes, given the requirement that the code dimensions sum to less than the length $n$; in fact, the sum must be $\leq(1-\Omega(1))n$. Indeed, \cite{kalachev_two-sided_2023} observe that this requirement is inherent to Reed-Solomon codes, and that as a result, such codes cannot be used as the local codes in asymptotically good balanced product codes with current analysis techniques, which require product expansion for a pair of codes $(C_1,C_2)$ as well as their duals $(C_1^\perp,C_2^\perp)$. Note that if $\dim(C_1)+\dim(C_2)<n$, then $\dim(C_1^\perp)+\dim(C_2^\perp)>n$. We are able to avoid such a barrier in Theorem~\ref{thm:sspeinf}, as for instance in the $t=2$ case, we only need product-expansion for $({Q_X^1}^\perp,Q_Z^2)$ and for $({Q_Z^1}^\perp,Q_X^2)$. Thus as long as we choose $Q_X$ to have larger dimension (by at least $\Omega(n)$) than $Q_Z$, then we will have $\dim({Q_X^1})^\perp+\dim(Q_Z^2)$ and $\dim({Q_Z^1})^\perp+\dim(Q_X^2)$ both at most $(1-\Omega(1))n$, and we can apply the product-expansion bound of \cite{polishchuk_nearly-linear_1994}. It is an interesting question whether such an approach can somehow extend to the analysis of balanced product codes with local Reed-Solomon codes.

We now provide some details on the proof of Theorem~\ref{thm:sspeinf}. The formula for the dimension $K$ in Theorem~\ref{thm:sspeinf} follows directly from the well-known K\"{u}nneth formula; a variant for single-sector complexes was shown by \cite{bravyi_homological_2014} (see Proposition~\ref{prop:sskunneth}). Thus our key contribution in Theorem~\ref{thm:sspeinf} are the bounds on distance and soundness. Below, we prove the $t=2$ case of the distance bound. The proof of the bound for general $t$, as well as the soundness bound, use similar ideas, and can be found in Section~\ref{sec:sspeproof}.

\begin{proof}[Proof of distance bound in $t=2$ case of Theorem~\ref{thm:sspeinf}]
  Consider an arbitrary $a\in Q_Z\setminus Q_X^\perp$. We will show that $|a|\geq\rho^2n_1n_2$; an analogous proof applies to $a\in Q_X\setminus Q_Z^\perp$, so the desired bound $D\geq\rho^2n_1n_2$ follows.

  By the K\"{u}nneth formula (Proposition~\ref{prop:sskunneth}), there exists a decomposition $a=a_0+a_1+a_2$ for some $a_0\in Q^1_Z\otimes Q^2_Z$, $a_1\in({Q^1_X}^\perp)^{(1)}$, $a_2\in({Q^2_X}^\perp)^{(2)}\subseteq(Q^2_Z)^{(2)}$. Furthermore, there exist some $z^1\in Q^1_X$, $z^2\in Q^2_X$ such that $(z^1\otimes z^2)\cdot a_0=({z^1}^\top\otimes{z^2}^\top)a_0\neq 0$. Indeed, this latter statement follows from the K\"{u}nneth formula along with the nondegeneracy of the natural cohomology/homology bilinear form; see Lemma~\ref{lem:hombasis}, Corollary~\ref{cor:ssprodbases}, and Corollary~\ref{cor:ssprodcycles} for details. Put in the language of quantum codes, the statement says that $a_0$ is a nontrivial logical $Z$ operator, so it anticommutes with some element of the basis of logical $X$ operators given by the K\"{u}nneth formula.

  Because $({Q^1_X}^\perp,Q^2_Z)$ has product-expansion $\geq\rho$, there exist $a_1'\in({Q^1_X}^\perp)^{(1)}$, $a_2'\in (Q^2_Z)^{(2)}$ such that $a_1'+a_2'=a_1+(a_0+a_2)=a$ and
  \begin{equation}
    \label{eq:aboundinf}
    |a| \geq \rho(n_1|a_1'|_1+n_2|a_2'|_2).
  \end{equation}

  By definition, $a_2'-(a_0+a_2)=a_1-a_1'\in ({Q^1_X}^\perp)^{(1)}\cap(Q^2_Z)^{(2)}={Q^1_X}^\perp\otimes Q^2_Z$. Thus
  \begin{equation*}
    (H^1_Z\otimes{z^2}^\top)a_2' = (H^1_Z\otimes{z^2}^\top)(a_0+a_2)+(H^1_Z\otimes{z^2}^\top)(a_2'-(a_0+a_2)) = 0,
  \end{equation*}
  where $(H^1_Z\otimes I)a_0=0$ because $a_0\in Q^1_Z\otimes Q^2_Z$, $(I\otimes{z^2}^\top)a_2=0$ because $a_2\in({Q^2_X}^\perp)^{(2)}$ with $z^2\in Q^2_X$, and $(H^1_Z\otimes I)(a_2'-(a_0+a_2))=0$ because $a_2'-(a_0+a_2)\in{Q^1_X}^\perp\otimes Q^2_Z$. It follows that $(I\otimes{z^2}^\top)a_2'\in Q^1_Z$. But we also have that
  \begin{equation*}
    ({z^1}^\top\otimes{z^2}^\top)a_2' = ({z^1}^\top\otimes{z^2}^\top)((a_0+a_2)+(a_2'-(a_0+a_2))) = ({z^1}^\top\otimes{z^2}^\top)a_0 \neq 0,
  \end{equation*}
  where we again use the fact that $(I\otimes{z^2}^\top)a_2=0$, and also that $({z^1}^\top\otimes{z^2}^\top)(a_2'-(a_0+a_2))=0$ because $a_2'-(a_0+a_2)\in{Q^1_X}^\perp\otimes Q^2_Z$ with $z^1\in Q^1_X$. Thus ${z^1}^\top(I\otimes{z^2}^\top)a_2'\neq 0$, so we must have $(I\otimes{z^2}^\top)a_2'\neq 0$

  We have now shown that $(I\otimes{z^2}^\top)a_2'\in Q^1_Z\setminus\{0\}$. But by assumption in the theorem statement, $Q^1_Z$ has product-expansion $\geq\rho$, and the product-expansion of a single code equals its relative distance. Thus $|(I\otimes{z^2}^\top)a_2'|\geq\rho n_1$. Therefore must be at least $\rho n_1$ nonzero rows in the $n_1\times n_2$ matrix $a_2$; that is, $|a_2'|_2\geq\rho n_1$. It then follows by~(\ref{eq:aboundinf}) that $|a|\geq\rho^2n_1n_2$, as desired.
\end{proof}

\subsection{Homological Products of Subsystem QLDPC Codes}
\label{sec:subprodinf}
In this section, we describe our results on homological products of chain complexes associated to quantum LDPC and locally testable codes. For these results, we work with homological products over the more standard multi-sector chain complexes, which we simply call chain complexes:

\begin{definition}
  A \textbf{(multi-sector) chain complex $\cC_*$ over $\bF_q$} consists of a sequence of $\bF_q$-vector spaces $(C_i)_{i\in\bZ}$ and linear \textbf{boundary maps} $(\partial^{\cC}_i:C_i\rightarrow C_{i-1})_{i\in\bZ}$ satisfying $\partial^{\cC}_{i-1}\partial^{\cC}_i=0$ for all $i\in\bZ$. The associated \textbf{cochain complex $\cC^*$} has vector spaces $C^i=C_i$ and coboundary maps $\delta^{\cC}_i=(\partial^{\cC}_{i+1})^\top$. The \textbf{locality} of $\cC_*$ is the maximum Hamming weight of any row or column of any $\partial^{\cC}_i$. The \textbf{$i$-homology (resp.~$i$-cohomology) group} of $\cC_*$ is $H_i(\cC)=\ker(\partial^{\cC}_i)/\im(\partial^{\cC}_{i+1})$ (resp.~$H^i(\cC)=\ker(\delta^{\cC}_i)/\im(\delta^{\cC}_{i-1})$).

  For chain complexes $\cA,\cB$, the \textbf{homological product $\cC+\cA\otimes\cB$} is the chain complex given by $C_i=\bigoplus_{j\in\bZ}A_j\otimes B_{i-j}$ and $\partial^{\cC}_i=\bigoplus_{j\in\bZ}(\partial^{\cA}_j\otimes I+(-1)^jI\otimes\partial^{\cB}_{i-j})$.
\end{definition}

We typically work with chain complexes $\cC$ that have $C_i=0$ for every $i$ outside of some sequence of $t$ consecutive integers; we then call $\cC$ a $t$-term chain complex. For every level $i$ with $C_i\neq 0$, there is a naturally associated quantum CSS code $Q$ given by $Q_X=\ker\delta^{\cC}_i$ and $Q_Z=\ker\partial^{\cC}_i$. Conversely, given a length-$n$ CSS code $Q=(Q_X,Q_Z)$ with parity-check matrices $H_X\in\bF_q^{m_X\times n},\;H_Z\in\bF_q^{m_Z\times n}$, we can construct a 3-term chain complex $\cC=(\bF_q^{m_X}\xrightarrow{H_X^\top}\bF_q^n\xrightarrow{H_Z}\bF_q^{m_Z})$ with associated quantum code $Q$ at the middle level. By definition, the maximum stabilizer weight of $Q$ is always at most the locality of $\cC$. The dimension (i.e.~number of logical qudits) of $Q$ equals $\dim(H^i(\cC))=\dim(H_i(\cC))$, and elements of $H^i(\cC)$ and $H_i(\cC)$ can be viewed as equivalence classes of $X$ and $Z$ logical operators, respectively, up to stabilizers. The distance of $Q$ equals the minimum Hamming weight of any representative of a nontrivial element of $H^i(\cC)$ or $H_i(\cC)$.

For $\cC=\cA\otimes\cB$, the well-known K\"{u}nneth formula (see Proposition~\ref{prop:kunneth}) provides an isomorphism
\begin{equation*}
  H_i(\cC) \cong \bigoplus_{j\in\bZ}H_j(\cA)\otimes H_{i-j}(\cB),
\end{equation*}
which for $a\in\ker(\partial^{\cA}_j)$ and $b\in\ker(\partial^{\cB}_{i-j})$ is given by
\begin{equation*}
  a\otimes b+\im(\partial^{\cC}_i) \mapsfrom (a+\im(\partial^{\cA}_j))\otimes(b+\im(\partial^{\cB}_{i-j})).
\end{equation*}
The analogous result naturally also applies to cohomology.

\subsubsection{The $\sqrt{N}$ Distance Barrier}
\label{sec:sqrtbarrier}
Our goal in this paper is to characterize conditions under which homological products preserve good distance (i.e.~close to linear in the code length) of the associated quantum codes. However, there is a somewhat fundamental difficulty that arises when we work with finite-length complexes. Recall that for chain complexes $\cA,\cB$ and integers $i,j$, the K\"{u}nneth formula implies that if $a\in\ker(\partial^{\cA}_j)\setminus\im(\partial^{\cA}_{j+1})$ and $b\in\ker(\partial^{\cB}_{i-j})\setminus\im(\partial^{\cB}_{i-j+1})$, then $a\otimes b\in\ker(\partial^{\cC}_i)\setminus\im(\partial^{\cC}_{i+1})$. Therefore for instance consider the greatest $j\in\bZ$ with $A_j\neq 0$, so that $\partial^{\cA}_j=0$, and assume that $B_{i-j}\neq 0$ (this assumption is without loss of generality if we allow ourselves to swap $\cA$ and $\cB$, or instead consider their cochain complexes if necessary). Then unless $\partial^{\cA}_{j+1}$ is surjective (which can be difficult to enforce in many settings), we must have some $a\in\ker(\partial^{\cA}_j)\setminus\im(\partial^{\cA}_{j+1})$ of Hamming weight $|a|=1$. Therefore $|a\otimes b|=|b|$, which (assuming $\dim(A_j)\approx\dim(B_{i-j})$) will be at most roughly the square root of the length of the the quantum code at level $i$ of $\cC$.

This phenomenon was a major reason behind the ``$\sqrt{N}$ barrier,'' as prior to the work of \cite{hastings_fiber_2021} it was a longstanding open question to construct length-$N$ qLDPC codes with distance greater than $\tilde{\Theta}(\sqrt{N})$. This barrier is perhaps easiest to understand with the homological product of two 2-term complexes $\cA,\cB$ (i.e.~the ``hypergraph product'' of \cite{tillich_quantum_2014}), which can be visualized as a $2\times 2$ grid of vector spaces with maps between them:
\begin{equation}
  \label{eq:hgpcd}
  \begin{tikzcd}
    A_1\otimes B_1 \arrow[r,"-I\otimes\partial_1^{\cB}"] \arrow[d,"\partial_1^{\cA}\otimes I"] & A_1\otimes B_0 \arrow[d,"\partial_1^{\cA}\otimes I"] \\
    A_0\otimes B_1 \arrow[r,"I\otimes\partial_1^{\cB}"] & A_0\otimes B_0
  \end{tikzcd}
\end{equation}
The diagram above represents the 3-term chain complex $\cC=\cA\otimes\cB$, where level $2$ is the top left corner $A_1\otimes B_1$, level $1$ contains the diagonal elements $A_0\otimes B_1\oplus A_1\otimes B_0$, and level $0$ is the bottom right corner $A_0\otimes B_0$. The boundary maps of $\cC$ are simply the maps shown in the diagram. We let $Q$ be the quantum code at level $1$ of $\cC$.

Assume for simplicity that all $A_0,A_1,B_0,B_1$ all have dimension $\Theta(n)$. If $Q$ is a code of positive dimension, the K\"{u}nneth formula implies that either $\partial^{\cA}_1$ is not surjective and $\partial^{\cB}_1$ is not injective, or vice versa; assume the former, as the argument for the latter is analogous. Then as described above, there is some $a\in A_0\setminus\im(\partial^{\cA})$ of weight $|a|=1$, and there is some $b\in\ker(\partial^{\cB}_1)$, so $a\otimes b\in\ker(\partial^{\cC}_1)\setminus\im(\partial^{\cC}_2)=Q_Z\setminus Q_X^\perp$ has Hamming weight $|a\otimes b|=|b|\leq O(n)$. But $Q$ is a code of length $N=\Theta(n^2)$, and thus $Q$ has distance $D\leq O(\sqrt{N})$.

Thus homological products of 2-term chain complexes cannot yield codes of good distance. We could consider using complexes $\cA,\cB$ with $t^{\cA},t^{\cB}\geq 2$ terms, and letting $Q$ be the quantum code at some level $i$ of the product complex $\cC=\cA\otimes\cB$. However, a similar issue will arise. Specifically, the $2\times 2$ grid in~(\ref{eq:hgpcd}) would be replaced by a $t^{\cA}\times t^{\cB}$ grid, but if we look at vector spaces around the edge of this grid, we may again find codewords in $\ker(\partial^{\cC}_i)\setminus\im(\partial^{\cC}_{i+1})=Q_Z\setminus Q_X^\perp$ of the form $a\otimes b$ with $|a|=1$, so that $|a\otimes b|=|b|\leq O(n)=O(\sqrt{N})$.

\cite{bravyi_homological_2014} circumvented this barrier using single-sector complexes, which can also be viewed as infinitely long complexes $\cdots\xrightarrow{\partial}C\xrightarrow{\partial}C\xrightarrow{\partial}\cdots$ and hence do not have such low-distance codewords in vector spaces on the edges of the grid described above. However, single-sector complexes are somewhat restrictive, as they only have a single boundary map, and we were not able to obtain good bounds on the distance of single-sector homological products of qLDPC codes with constant locality. Instead, we took an alternative approach to circumventing this barrier, using subsystem codes as described below.

\subsubsection{Circumventing the Barrier with Subsystem Codes}
\label{sec:aroundbarrier}
Our key observation for circumventing the $\sqrt{N}$ distance barrier described above is that it only applies to certain subspaces of the code space. For instance, consider a homological product of 3-term chain complexes $\cA,\cB$ as a $3\times 3$ grid:
\begin{equation}
  \label{eq:homcd}
  \begin{tikzcd}
    A_2\otimes B_2  \arrow[r,"I\otimes\partial_2^{\cB}"] \arrow[d,"\partial_2^{\cA}\otimes I"] & A_2\otimes B_1 \arrow[r,"I\otimes\partial_1^{\cB}"] \arrow[d,"\partial_2^{\cA}\otimes I"] & A_2\otimes B_0 \arrow[d,"\partial_2^{\cA}\otimes I"] \\
    A_1\otimes B_2  \arrow[r,"-I\otimes\partial_2^{\cB}"] \arrow[d,"\partial_1^{\cA}\otimes I"] & A_1\otimes B_1 \arrow[r,"-I\otimes\partial_1^{\cB}"] \arrow[d,"\partial_1^{\cA}\otimes I"] & A_1\otimes B_0 \arrow[d,"\partial_1^{\cA}\otimes I"] \\
    A_0\otimes B_2 \arrow[r,"I\otimes\partial_2^{\cB}"] & A_0\otimes B_1 \arrow[r,"I\otimes\partial_1^{\cB}"] & A_0\otimes B_0
  \end{tikzcd}
\end{equation}
Let $Q$ be the quantum code at the middle level ($i=2$) of the 5-term chain complex $\cC=\cA\otimes\cB$, so that $Q_X,Q_Z$ are codes within the vector space $C_2=A_0\otimes B_2\oplus A_1\otimes B_1\oplus A_2\otimes B_0$ consisting of the three spaces in the bottom-left to top-right diagonal of the grid in~(\ref{eq:homcd}). By the K\"{u}nneth formula,
\begin{equation*}
  H_2(\cC) \cong H_0(\cA)\otimes H_2(\cB) \oplus H_1(\cA)\otimes H_1(\cB) \oplus H_2(\cA)\otimes H_0(\cB),
\end{equation*}
where for simplicity here we restrict attention to homology (i.e.~$Z$ logical operators); the analogous reasoning also applies to cohomology (i.e.~$X$ logical operators). The reasoning in Section~\ref{sec:sqrtbarrier} suggests that elements of $H_0(\cA)\otimes H_2(\cB)$ and $H_2(\cA)\otimes H_0(\cB)$ may have low-weight representatives. However, there is no apparent obstruction to having the homology classes associated to nonzero elements of $H_1(\cA)\otimes H_1(\cB)$ having only high-weight representatives (and analogously for cohomology). In the language of quantum codes, while the code $Q$ at level $2$ of $\cC$ may have low-weight logical operators, we may hope to guarantee that the subspace of logical operators associated to the middle element of the grid in (\ref{eq:homcd}) all have high weight. In Theorem~\ref{thm:proddisinf} below, we realize this hope for $\cA,\cB$ satisfying appropriate conditions.

This idea of a code where where we restrict attention to a subspace of logical operators has been studied in the quantum error correction literature, where it is called a \textit{subsystem code}, defined below.

\begin{definition}
  \label{def:subsysteminf}
  A \textbf{quantum subsystem code} is a pair $(Q,L)$ consisting of a CSS code $Q=(Q_X,Q_Z)$ together with a pair $L=(L_X,L_Z)$ of tuples $L_X=(c^1,\dots,c^k)\in Q_X^k$, $L_Z=(c_1,\dots,c_k)\in Q_Z^k$ of the same length $k$, such that $c^i\cdot c_j=\1_{i=j}$ for all $i,j\in[k]$.

  The dimension of the subsystem code $(Q,L)$ is $k$, the distance is
  \begin{equation*}
    d := \min_{c\in(Q_X\setminus\spn\{L_Z\}^\perp)\cup(Q_Z\setminus\spn\{L_X\}^\perp)}|c|,
  \end{equation*}
  and the locality $w$ equals the the locality (i.e.~stabilizer weight) of $Q$.
\end{definition}

\begin{remark}
  In Definition~\ref{def:subsysteminf}, we define the locality of a subsystem code to equal the maximum stabilizer weight for the chosen parity-check matrices of the underlying CSS code $Q$. Therefore for our purposes, a subsystem code with constant locality is simply a qLDPC code with constant-weight stabilizers, but where we only encode message qudits into certain logical operators. Such codes can typically be used for fault-tolerant computation analogously to non-subsystem qLDPC codes.

  In contrast, some works (e.g.~\cite{bacon_sparse_2017}) have studied subsystem codes where the locality is measured by the weight of generators for the groups $Q_X\cap\spn\{L_Z\}^\perp$ and $Q_Z\cap\spn\{L_X\}^\perp$, which are often called ``gauge operators.'' Codes with such a weaker locality property can be more difficult to use for applications to fault-tolerance, as while a high-weight stabilizer can be measured by decomposing it into low-weight gauge operators, many different gauge operators may need to be measured for each stabilizer.
\end{remark}

\subsubsection{Results on Products of QLDPC Subsystem Codes}
\label{sec:qldpcsubresults}
In this section, we state our main results on products of qLDPC codes. In Theorem~\ref{thm:proddisinf} below, we characterize conditions under which good subsystem qLDPC code distance is preserved by homological products. Following the intuition described in Section~\ref{sec:aroundbarrier} above, we restrict attention to logical operators in the middle terms of homological products of complexes of length $\geq 3$. In fact, we take the product of a 3-term complex $\cA$ with a 5-term complex $\cB$; for both complexes, we assume that the middle level is indexed $i=0$.

\begin{theorem}[Informal statement of Theorem~\ref{thm:proddis}]
  \label{thm:proddisinf}
  Let $\cA$ be a 3-term chain complex of locality $w^{\cA}$ with an associated $[[n^{\cA},k^{\cA},d^{\cA}]]$ subsystem code at the middle level ($i=0$). Let $\cB$ be a 5-term chain complex of locality $w^{\cB}$ with an associated $[[n^{\cB},k^{\cB},d^{\cB}]]$ CSS code a the middle level ($i=0$) that is locally testable with soundness $\rho^{\cB}$. Then the homological product $\cC=\cA\otimes\cB$ is a 7-term chain complex of locality $w^{\cC}\leq w^{\cA}+w^{\cB}$ with an associated $[[N,K,D]]$ subsystem code at the middle level for
  \begin{align*}
    N &\leq n^{\cA}n^{\cB}+2{n'}^{\cA}{n'}^{\cB} \\
    K &= k^{\cA}k^{\cB} \\
    D &\geq \frac{d^{\cA}d^{\cB}}{(w^{\cA}/\rho^{\cB})\cdot(n^{\cB}/{d'}^{\cB})},
  \end{align*}
  where we assume that the CSS codes at levels $i=\pm 1$ of $\cB$ (resp.~$\cA$) have length $\leq{n'}^{\cB}$ (resp.~$\leq{n'}^{\cA}$) and distance $\geq{d'}^{\cB}$.
\end{theorem}

Theorem~\ref{thm:proddisinf} provides a way to obtain good larger qLDPC codes from smaller ones. For instance, if ${n'}^{\cA},k^{\cA},d^{\cA}=\Theta(n^{\cA})$, ${n'}^{\cB},k^{\cB},d^{\cB},{d'}^{\cB}=\Theta(n^{\cB})$, and if $w^{\cA},w^{\cB},\rho^{\cB}$ are all bounded by constants, so that $\cA$ and $\cB$ have asymptotically good parameters, then Theorem~\ref{thm:proddisinf} provides a subsystem qLDPC code of length $N=\Theta(n^{\cA}n^{\cB})$ associated to the product $\cC=\cA\otimes\cB$ that also has asymptotically good parameters.

Under stronger conditions on $\cA$ and $\cB$, we can also bound the local testability of the product $\cC$; see Theorem~\ref{thm:prodfill} in Section~\ref{sec:genbounds} for details. A similar result also holds for complexes with arbitrarily many terms, though for simplicity in this paper we use the minimum length complexes that suffice for our purposes. As discussed below, a proof of a similar result for longer complexes can be found in \cite{kaufman_new_2021}.

The proof of Theorem~\ref{thm:proddisinf} goes by a diagram-chasing argument inspired by \cite{kaufman_new_2021}. Specifically, \cite{kaufman_new_2021} proved a related result, which they applied to show good $Z$-distance (i.e.~high weight for logical $Z$ operators) for quantum codes obtained as homological products of chain complexes from high-dimensional expanders \cite{lubotzky_ramanujan_2005,lubotzky_explicit_2005}.
However, \cite{kaufman_new_2021} did not use subsystem codes, and the $X$-distance of their product codes was only polylogarithmic in the block length.\footnote{Such poor $X$-distance was in part due to the fact that the cochain complexes they considered from high-dimensional expanders have much better distance and expansion properties than the associated chain complexes. However, as \cite{kaufman_new_2021} did not use subsystem codes, it seems difficult to obtain product codes with close-to-linear distance from their results, even when applied to products of different complexes.}
Thus after applying a distance-balancing procedure, \cite{kaufman_new_2021} obtained qLDPC codes of distance $\sqrt{N}\cdot\poly(\log N)$, which only barely surpasses the $\sqrt{N}$ barrier.

If both $\cA$ and $\cB$ have associated codes with large distance, then Theorem~\ref{thm:proddisinf} yields a large distance bound on the subsystem product code whenever $\cA$ has good (small) locality $w^{\cA}$ and $\cB$ has good (large) soundness $\rho^{\cB}$. At a high level, our proof of Theorem~\ref{thm:proddisinf} (inspired by \cite{kaufman_new_2021}) argues that by the K\"{u}nneth formula, codewords of the product code can be expressed as tensor products of codewords of $\cA$ and $\cB$, plus elements in the image of the boundary map $\partial^{\cC}$ (i.e.~stabilizers). Classical tensor product codes have good distance, so we just need to argue that the distance is not significantly reduced by adding in elements in $\im(\partial^{\cC})$. This final step goes by a diagram-chasing argument that bounds the loss in distance in terms of $w^{\cA}$ and $\rho^{\cB}$; the reader is referred to Section~\ref{sec:genbounds} for the more details.

In Theorem~\ref{thm:iterativeinf} below, we iteratively apply Theorem~\ref{thm:proddisinf} starting with a constant-sized chain complex from \cite{dinur_expansion_2024,kalachev_personal_2024} to obtain an infinite family of qLDPC codes with close-to-linear dimension and distance. Below, we view qLDPC codes $Q$ interchangeably with their associated 3-term chain complexes, as described previously.

\begin{theorem}[Informal statement of Theorem~\ref{thm:iterative}]
  \label{thm:iterativeinf}
  For every $\epsilon>0$, there exists a constant-sized 5-term chain complex $\cC$ and an infinite explicit sequence of qLDPC subsystem codes $(Q_i,L_i)_{i\in\bN}$ with parameters
  \begin{equation*}
    [[N_i,\;K_i\geq N_i^{1-\epsilon},\;D_i\geq N_i^{1-\epsilon}]]_2
  \end{equation*}
  and locality $w_i=O(1)$, such that each $Q_i$ is obtained by applying the stabilizer-weight-reduction transformation of \cite{hastings_quantum_2023} to the middle $3$ terms of the the homological product $Q_{i-1}\otimes\cC$.
\end{theorem}

The weight-reduction transformation of \cite{hastings_quantum_2023} provides a means of reducing the locality of a CSS code from a given value $w$ to an absolute constant $O(1)$, at the cost of increasing the code length and decreasing the distance by a factor of $\poly(w)$; see Theorem~\ref{thm:locred} for a formal statement. We apply this transformation after every application of the homological product because the distance bound in Theorem~\ref{thm:proddisinf} decays with the locality of $\cA$. Therefore by maintaining constant locality, we are able to inductively maintain a good distance bound in each iteration.

The starting complex in Theorem~\ref{thm:iterativeinf} is a constant-sized instance of the construction of \cite{dinur_expansion_2024}, which applies a forthcoming result of \cite{kalachev_personal_2024} to obtain chain complexes whose associated CSS codes are quantum LTCs with nearly optimal asymptotic parameters. As Theorem~\ref{thm:iterativeinf} only needs a constant-sized such starting object $\cC$, it is an interesting question to instantiate it with alternative constructions.



\section{Preliminaries}
\label{sec:prelim}
This section introduces necessary notation and preliminary notions.

\subsection{Notation}
For a positive integer $n$, we let $[n]=\{1,\dots,n\}$. For a prime power $q$, we let $\bF_q$ denote the finite field of order $q$. For a vector $v\in\bF_q^n$, we let $|v|=|\{i\in[n]:v_i\neq 0\}|$ denote the Hamming weight of $v$, meaning the number of nonzero coordinates in the standard (or otherwise specified) basis. Unless explicitly stated otherwise, we will assume all $\bF_q$-vector spaces that we consider are finite-dimensional, so that all Hamming weights are finite.

Sometimes it will be helpful to consider more general Hamming weights, in which we partition the coordinates into groups and count the number of distinct groups containing a nonzero coordinate; the ordinary Hamming weight is recovered by choosing the finest partition. For instance, for a matrix $c\in\bF_q^{n_1\times n_2}$, we let $|c|_1$ (resp.~$|c|_2$) denote the number of nonzero columns (resp.~rows) of $c$. However, we will always specify such notation before using it.

Elements of $\bF_q^{n_1}\otimes\cdots\otimes\bF_q^{n_t}=\bF_q^{n_1\times\cdots\times n_t}$ are sometimes referred to as $t$-dimensional tensors. For $i\in[t]$, a direction-$i$ column in such a tensor is a vector in $\bF_q^{n_i}$ obtained by restricting the tensor to a set of coordinates $(j_1,\dots,j_t)$ in which $j_k$ is fixed for all $k\neq i$, and $j_i$ takes all possible values in $[n_i]$.

\subsection{Classical Codes}
This section describes basic notions in classical coding theory.

\begin{definition}
  For a finite field $\bF_q$, a \textbf{classical linear code of length $n$ and dimension $k$ over $\bF_q$} is a $k$-dimensional linear subspace $C\subseteq\bF_q^n$. The \textbf{rate} of $C$ is $R:=k/n$. The \textbf{distance} $d$ of $C$ is the minimum Hamming weight of a nonzero element of $C$, that is $d=\min_{c\in C\setminus\{0\}}|c|$. We summarize the above parameters by saying that $C$ is an $[n,k,d]_q$ code. The \textbf{dual code} $C^\perp\subseteq\bF_q^n$ of $C$ is defined by $C^\perp=\{x\in\bF_q^n:x\cdot y=0\;\forall y\in C\}$, where $x\cdot y=\sum_{i\in[n]}x_iy_i$ denotes the standard bilinear form.
\end{definition}

All classical codes we consider in this paper will be linear, so we will often simply say ``classical code'' or even (when clear from context) ``code'' to refer to a classical linear code.

\begin{definition}
  \label{def:classtensor}
  For classical codes $C_1\subseteq\bF_q^{n_1}$ and $C_2\subseteq\bF_q^{n_2}$, the \textbf{tensor code} $C_1\otimes C_2\subseteq\bF_q^{n_1}\otimes\bF_q^{n_2}=\bF_q^{n_1\times n_2}$ consists of all $n_1\times n_2$ matrices for which every column lies in $C_1$ and every row lies in $C_2$. For $c_1\in C_1,c_2\in C_2$, we let $c_1\otimes c_2\in C_1\otimes C_2$ be the tensor codeword given by $(c_1\otimes c_2)_{i,j}=(c_1)_i(c_2)_j$. We furthermore define the \textbf{dual tensor code} $C_1\boxplus C_2\subseteq\bF_q^{n_1}\otimes\bF_q^{n_2}$ by
  \begin{equation*}
    C_1\boxplus C_2 = (C_1^\perp\otimes C_2^\perp)^\perp = C_1\otimes\bF_q^{n_2}+\bF_q^{n_1}\otimes C_2.
  \end{equation*}
\end{definition}

Note that ``dual tensor'' is a slight abuse of terminology that we use for conciseness, as the phrase ``dual tensor of duals'' would be more accurate.

We will use the following well-studied class of codes:

\begin{definition}
  \label{def:RS}
  Given a prime power $q$ and a positive integer $k$, the \textbf{Reed-Solomon code} $\textrm{RS}(q,k)$ is the classical code of length $q$ and dimension $k$ over $\bF_q$ consisting of the evaluations of all univariate polynomials of degree $<k$ on all points in $\bF_q$. That is,
  \begin{equation*}
    \textrm{RS}(q,k) = \left\{(f(x))_{x\in\bF_q}\in\bF_q^{\bF_q}\cong\bF_q^q:f(X)\in\bF_q[X],\;\deg(f)<k\right\}.
  \end{equation*}
\end{definition}

The following lemma is well-known:
\begin{lemma}
  \label{lem:RSdual}
For all integers $k$, $0 \le k \le q$,  $\textrm{RS}(q,k)^\perp = \textrm{RS}(q,q-k)$.
\end{lemma}

Definition~\ref{def:RS} can be generalized to allow for evaluations at only a subset of the points in $\bF_q$; Lemma~\ref{lem:RSdual} also has a corresponding generalization. However, for simplicity in this paper we restrict attention to the case where the polynomials are evaluated on the whole field.

\subsection{Quantum Codes}
This section describes basic notions in quantum coding theory. In this paper, we restrict attention to quantum CSS codes, described below.

\begin{definition}
  \label{def:CSScode}
  For a finite field $\bF_q$, a \textbf{quantum CSS code of length $n$ over $\bF_q$} is a pair $Q=(Q_X,Q_Z)$ of subspaces (i.e.~classical codes) $Q_X,Q_Z\subseteq\bF_q^n$ such that $Q_X^\perp\subseteq Q_Z$. The \textbf{dimension} of $Q$ is $k:=\dim(Q_Z)-\dim(Q_X^\perp)$, and the \textbf{rate} is $R:=k/n$. The \textbf{distance} of $Q$ is
  \begin{equation*}
    d := \min_{c\in(Q_X\setminus Q_Z^\perp)\cup(Q_Z\setminus Q_X^\perp)}|c|.
  \end{equation*}
  We summarize the above parameters by saying that $Q$ is an $[[n,k,d]]_q$ code.

  The quantum code $Q$ is \textbf{low-density parity-check (LDPC) of locality $w$} if there exist parity-check matrices $H_X\in\bF_q^{m_X\times n}$, $H_Z\in\bF_q^{m_Z\times n}$ for $Q_X,Q_Z$ respectively, meaning that $Q_X=\ker H_X$ and $Q_Z=\ker H_Z$, such that every row and column of $H_X$ and $H_Z$ has at most $w$ nonzero entries.

  The quantum code $Q$ is a \textbf{locally testable code (LTC) of soundness $\rho$} if the parity-check matrices $H_\alpha$, $\alpha=X,Z$ furthermore satisfy the following: for every $s\in\im H_\alpha$, there exists some $e\in\bF_q^n$ with $H_\alpha e=s$ such that
  \begin{equation*}
    \frac{|s|}{m_\alpha} \geq \rho\cdot\frac{|e|}{n}.
  \end{equation*}
\end{definition}

In this paper, we often simply say ``quantum code'' to refer to a quantum CSS code. We are interested in constructing families of quantum LDPC and locally testable codes with locality $w$ significantly less than the block length $n$. Typically, the study of qLDPC codes focuses on constructing families of codes with constant locality $w$ as the block length $n$ grows. However, we remark that \cite{hastings_quantum_2023,wills_tradeoff_2024} gives a method for reducing the locality $w$ to a constant at the cost of decreasing the code's distance and soundness. Thus quantum LDPC and locally testable codes of constant locality can be obtained from codes of sufficiently low (but growing) locality.

\begin{definition}
  An infinite family of classical or quantum codes is \textbf{asymptotically good} if the codes' dimension $k=\Theta(n)$ and distance $d=\Theta(n)$ both grow linearly as the block length $n\rightarrow\infty$. The family is \textbf{LDPC} if the locality $w=O(1)$ remains bounded by a constant.
\end{definition}

We will use the following quantum analogue of Definition~\ref{def:RS}:

\begin{definition}
  \label{def:QRS}
  A \textbf{quantum Reed-Solomon code} is a quantum CSS code $Q=(Q_X,Q_Z)$ such that both $Q_X$ and $Q_Z$ are classical Reed-Solomon codes in the sense of Definition~\ref{def:RS}.
\end{definition}

\subsection{Chain Complexes}
\label{sec:cc}
In this paper, we will use the language of chain complexes to describe the quantum codes we consider. In particular, this paper focuses on homological products of quantum codes, which are particularly natural to describe using chain complexes. Many definitions below have natural generalizations, such as to arbitrary rings, etc.~but for simplicity we only use a sufficiently general presentation for our purposes.

\begin{definition}
  \label{def:cc}
  A \textbf{chain complex $\cC_*$ over $\bF_q$} consists of a sequence of $\bF_q$-vector spaces $(C_i)_{i\in\bZ}$ and linear \textbf{boundary maps} $(\partial_i^{\cC}:C_i\rightarrow C_{i-1})_{i\in\bZ}$ satisfying $\partial_{i-1}^{\cC}\partial_i^{\cC}=0$ for all $i\in\bZ$. When clear from context, we omit the superscript and/or subscript and write $\partial=\partial_i=\partial^{\cC}=\partial_i^{\cC}$. Assuming that each $C_i$ has a fixed basis, then the \textbf{locality $w^{\cC}$} of $\cC$ is the maximum number of nonzero entries in any row or column of any matrix $\partial_i$ in this fixed basis. If there exist bounds $\ell<m\in\bZ$ such that for all $i<\ell$ and $i>m$ have $C_i=0$, then we may truncate the sequence and say that $\cC$ is the $(m-\ell+1)$-term chain complex
  \begin{equation}
    \label{eq:chaintruncdec}
    \cC_* = (C_m \xrightarrow{\partial_m} C_{m-1} \xrightarrow{\partial_{m-1}} \cdots \xrightarrow{\partial_{\ell+1}} C_\ell).
  \end{equation}
  We furthermore define the following (standard) vector spaces for $i\in\bZ$:
  \begin{align}
    \label{eq:cbhdefs}
    \begin{split}
      \text{the space of \bf $i$-cycles } Z_i(\cC) &:= \ker(\partial_i) \subseteq C_i \\
      \text{the space of \bf $i$-boundaries } B_i(\cC) &:= \im(\partial_{i+1}) \subseteq C_i \\
      \text{the \bf $i$-homology } H_i(\cC) &:= Z_i(\cC)/B_i(\cC)
    \end{split}
  \end{align}
  
  The \textbf{cochain complex $\cC^*$} associated to $\cC_*$ is the chain complex with vectors spaces $(C^i:=C_i)_{i\in\bZ}$ and boundary maps given by the \textbf{coboundary maps} $(\delta_i:=\partial_{i+1}^\top:C^i\rightarrow C^{i+1})_{i\in\bZ}$ obtained by transposing all the boundary maps of $\cC_*$. Thus the cochain complex of~(\ref{eq:chaintruncdec}) is
  \begin{equation*}
    \cC^* = (C^m \xleftarrow{\delta_{m-1}} C^{m-1} \xleftarrow{\delta_{m-2}} \cdots \xleftarrow{\delta_{\ell}} C^\ell).
  \end{equation*}
  We then analogously define the \textbf{$i$-cocycles $Z^i(\cC)=\ker(\delta_i)$, $i$-coboundaries $B^i(\cC)=\im(\delta_{i-1})$, and $i$-cohomology $H^i(\cC)=Z^i(\cC)/B^i(\cC)$}
\end{definition}

In this paper, we restrict attention to chain complexes in which all spaces $C_i$ are finite-dimensional $\bF_q$-vector spaces, so the following standard lemma holds.

\begin{lemma}[Well known]
  \label{lem:hombasis}
  Let $\cC$ be a chain complex over $\bF_q$ with all $C_i$ finite-dimensional. Then for all $i$,
  \begin{equation*}
    \dim H_i(\cC)=\dim H^i(\cC).
  \end{equation*}
  Furthermore, the bilinear form $\langle\cdot,\cdot\rangle:H^i(\cC)\times H_i(\cC)\rightarrow\bF_q$ given by
  \begin{equation*}
    \langle c'+B^i(\cC),c+B_i(\cC)\rangle:=c'\cdot c
  \end{equation*}
  is a well-defined nondegenerate form, independent of the choice of coset representatives $c',c$. In particular, letting $k_i=\dim H_i(\cC)$, then there exists a basis $c^1+B^i(\cC),\dots,c^{k_i}+B^i(\cC)$ for $H^i(\cC)$ and a basis $c_1+B_i(\cC),\dots,c_{k_i}+B_i(\cC)$ for $H_i(\cC)$ such that
  \begin{equation}
    \label{eq:dualbases}
    \langle c^j+B^i(\cC),c_\ell+B_i(\cC)\rangle=c^j\cdot c_\ell=\1_{j=\ell} \hspace{1em} \forall \; j,\ell\in[k_i].
  \end{equation}
\end{lemma}

Lemma~\ref{lem:hombasis} can be proven by basic linear-algebraic manipulations. For readers familiar with the language of quantum stabilizer codes, the lemma says that the space of logical operators of a CSS code can be decomposed into anticommuting $X$ and $Z$ operators; a proof can be found in Section~10.5.7 of~\cite{nielsen_quantum_2010}.

\begin{definition}
  Let $\cC$ be a chain complex with $k_i=\dim H_i(\cC)$. We say that vectors $c^1,\dots,c^{k_i}\in Z^i(\cC)$ and $c_1,\dots,c_{k_i}\in Z_i(\cC)$ form \textbf{dual bases for $i$-cohomology and $i$-homology} if they satisfy~(\ref{eq:dualbases}).
\end{definition}

Recall that by definition $\im(\partial_{i+1})=\ker(\partial_{i+1}^\top)^\perp$ and that $\delta_i=\partial_{i+1}^\top$. Therefore the chain complex condition $\partial_i\partial_{i+1}=0$ can be equivalently stated as saying that $\ker(\delta_i)^\perp\subseteq\ker(\partial_i)$, which is precisely the CSS condition in Definition~\ref{def:CSScode}. This observation motivates the following definition:

\begin{definition}
  \label{def:cctoqcode}
  For a chain complex $\cC$ and an integer $i\in\bZ$, the \textbf{quantum code associated to level $i$ of $\cC$} is the CSS code $Q=(Q_X,Q_Z)$ with $Q_X=\ker\delta_i$ and $Q_Z=\ker\partial_i$. We by default set $i=0$ and say that the \textbf{quantum code associated to $\cC$} is given by $Q_X=\ker\delta_0$ and $Q_Z=\ker\partial_0$.
\end{definition}

Definition~\ref{def:cctoqcode} associates cohomology with $Q_X/Q_Z^\perp$ and associates homology with $Q_Z/Q_X^\perp$. We maintain this convention throughout the paper, though note that in some other works these associations are swapped.

The following lemma is immediate from the above definitions:

\begin{lemma}
  Let $Q$ be the quantum code associated to level $i$ of a chain complex $\cC$. Then the dimension of $Q$ is $k=\dim H_i(\cC)=\dim H^i(\cC)$, and the locality of $Q$ is at most the locality of $\cC$.
\end{lemma}

We may similarly express distance and local testability in terms of chain complexes:

\begin{definition}
  \label{def:sysdisexp}
  For a chain complex $\cC$, the \textbf{$i$-systolic distance $d_i(\cC)$} and the \textbf{$i$-cosystolic distance $d^i(\cC)$} are defined as
  \begin{equation*}
    d_i(\cC) = \min_{c\in Z_i(\cC)\setminus B_i(\cC)}|c|, \hspace{2em} d^i(\cC) = \min_{c\in Z^i(\cC)\setminus B^i(\cC)}|c|.
  \end{equation*}
  The \textbf{$i$-filling constant $\mu_i(\cC)$} and the \textbf{$i$-cofilling constant $\mu^i(\cC)$} are defined as
  \begin{equation*}
    \mu_i(\cC) = \max_{b\in B_{i-1}(\cC)}\min_{c\in C_i:\partial_i(c)=b}\frac{|c|}{|b|}, \hspace{2em} \mu^i(\cC) = \max_{b\in B^{i+1}(\cC)}\min_{c\in C^i:\delta_i(c)=b}\frac{|c|}{|b|}.
  \end{equation*}
  The \textbf{$i$-collective filling constant $M_i(\cC)$} and the \textbf{$i$-collective cofilling constant $M^i(\cC)$} are defined as
  \begin{align*}
    M_i(\cC) &= \max_{m\in\bN}\;\max_{b_1,\dots,b_m\in B_{i-1}(\cC)}\;\min_{c_1,\dots,c_m\in C_i\;:\;\forall j\in[m],\;\partial_i(c_j)=b_j}\frac{\left|\bigcup_{j\in[m]}c_j\right|}{\left|\bigcup_{j\in[m]}b_j\right|} \\
    M^i(\cC) &= \max_{m\in\bN}\;\max_{b^1,\dots,b^m\in B^{i+1}(\cC)}\;\min_{c^1,\dots,c^m\in C^i\;:\;\forall j\in[m],\;\delta_i(c^j)=b^j}\frac{\left|\bigcup_{j\in[m]}c^j\right|}{\left|\bigcup_{j\in[m]}b^j\right|}.
  \end{align*}
\end{definition}

If a chain complex has small filling and cofilling constants at some level $i$, then high-weight errors on the associated quantum code must have large syndromes. Equivalently, low-weight syndromes can only arise from low-weight errors. This phenomenon is equivalent to local testability, as described below.

By definition, the distance $d(Q)$ of the quantum code $Q$ associated to level $i$ of $\cC$ is equal to
\begin{equation*}
  d(Q) = \min\{d_i(\cC),d^i(\cC)\}.
\end{equation*}
Similarly, letting $n_i=\dim C_i$, then the soundness $\rho(Q)$ (i.e.~local testability, see Definition~\ref{def:CSScode}) of $Q$ is equal to
\begin{equation*}
  \rho(Q) = \min\left\{\frac{n_i}{n_{i-1}}\cdot\frac{1}{\mu_i(\cC)},\;\frac{n_i}{n_{i+1}}\cdot\frac{1}{\mu^i(\cC)}\right\}.
\end{equation*}
That is, up to normalization, the (co)filling constants describe the local testability of $Q$. The reciprocols of the (co)filling constants are sometimes called \textbf{(co)cycle expansion constants}, though in this paper the (co)filling constants will be more convenient for our purposes.

The collective (co)filling constants are by definition at least as large as their noncollective counterparts. Therefore having small collective (co)filling constants is a \textit{stronger} property than having good local testability (i.e.~small (co)filling constants). This collective form of testability was introduced in \cite{kaufman_new_2021}, who showed that it is preserved under homological products, defined below.

\subsection{Homological product}
The above definitions motivate studying ways to construct larger chain complexes from smaller ones, as a way to construct larger quantum codes from smaller ones. The following definition describes one of the most natural such operations on chain complexes, namely a product.

\begin{definition}
  \label{def:homprod}
  For chain complexes $\cA$ and $\cB$, the \textbf{homological product $\cC=\cA\otimes\cB$} is the chain complex given by the vector spaces
  \begin{equation*}
    C_i := \bigoplus_{j\in\bZ}A_j\otimes B_{i-j}
  \end{equation*}
  and the boundary maps
  \begin{equation*}
    \partial^{\cC}_i := \bigoplus_{j\in\bZ}(\partial_j^{\cA}\otimes I+(-1)^jI\otimes\partial_{i-j}^{\cB})
  \end{equation*}
\end{definition}

The homological product in Definition~\ref{def:homprod}, along with its single-sector variant (see Section~\ref{sec:sscc} below), is the principal operation studied in this paper. The following result shows how the homology groups behave under such products.

\begin{proposition}[K\"{u}nneth formula (well-known, see e.g.~\cite{hatcher_algebraic_2001})]
  \label{prop:kunneth}
  Let $\cA$ and $\cB$ be chain complexes over a field $\bF_q$, each with a finite numbers of nonzero terms. Then for every $i\in\bZ$,
  \begin{equation*}
    H_i(\cA\otimes\cB) \cong \bigoplus_{j\in\bZ}H_j(\cA)\otimes H_{i-j}(\cB).
  \end{equation*}
  Furthermore, for $a\in Z_j(\cA)$ and $b\in Z_{i-j}(\cB)$, the isomorphism above maps\footnote{\label{footnote:Bclash} Here $B_i(\cA\otimes\cB)$ and $B_{i-j}(\cB)$ refer to spaces of boundaries, which are not to be confused with the vector spaces $B_i$ defining the chain complex $\cB$. Such a slight clash of notation will unfortunately reappear throughout the paper.}
  \begin{equation*}
    a\otimes b+B_i(\cA\otimes\cB) \mapsfrom (a+B_j(\cA))\otimes(b+B_{i-j}(\cB)).
  \end{equation*}
\end{proposition}

Thus the dimension of quantum codes behaves nicely under homological products. The goal of our work is to understand how the distance and local testability of quantum codes behave under such products.

\subsection{Subsystem Codes}
\label{sec:subsystem}
Sometimes it will be helpful to consider quantum codes in which we only encode logical qudits into certain degrees of freedom in the code. Such codes are called subsystem codes. Below we provide a definition that is sufficient for our purposes; more general definitions are also possible.

\begin{definition}
  \label{def:subsystem}
  A \textbf{quantum CSS subsystem code} $(Q,L)$ consists of a quantum CSS code $Q=(Q_X,Q_Z)$ together with a pair $L=(L_X,L_Z)$ of tuples $L_X=(c^1,\dots,c^k)\in Q_X^k$, $L_Z=(c_1,\dots,c_k)\in Q_Z^k$ of the same length $k$, such that $c^i\cdot c_j=\1_{i=j}$ for all $i,j\in[k]$.

  The \textbf{dimension} of the subsystem code $(Q,L)$ is $k$, and the \textbf{distance} is
  \begin{equation*}
    d := \min_{c\in(Q_X\setminus\spn\{L_Z\}^\perp)\cup(Q_Z\setminus\spn\{L_X\}^\perp)}|c|,
  \end{equation*}
  where $\spn\{L_X\}$ (resp.~$\spn\{L_Z\}$) refers to the span of the $k$ vectors in the tuple $L_X$ (resp.~$L_Z$).

  The \textbf{locality $w$} and \textbf{soundness $\rho$} of a subsystem code $(Q,L)$ are simply the corresponding parameters of the CSS code $Q$, as defined in Definition~\ref{def:CSScode}.
\end{definition}

\begin{remark}
  Definition~\ref{def:subsystem} provides a basis-dependent, and therefore slightly redundant, notion of subsystem code. In particular, if the tuples $L_X,L_Z$ are replaced with some other valid $L_X',L_Z'$ satisfying $\spn\{L_X,Q_Z^\perp\}=\spn\{L_X',Q_Z^\perp\}$ and $\spn\{L_Z,Q_X^\perp\}=\spn\{L_Z',Q_X^\perp\}$, then the resulting subsystem code $(Q,L')$ has the same properties as $(Q,L)$. This redundancy can be eliminated by defining gauge groups (see below), but we will not need such generality.

  Also note that in Definition~\ref{def:subsystem}, the requirement $c^i\cdot c_j=\1_{i=j}$ implies that all $c^i$ (resp.~$c_j$) lie in distinct cosets in $Q_X/Q_Z^\perp$ (resp.~$Q_Z/Q_X^\perp$), as if for instance $c^i-c^{i'}\in Q_Z^\perp$ for some $i\neq i'$, then we would have $1=(c^i-c^{i'})\cdot c_i=0$, a contradiction.
\end{remark}

For readers familiar with stabilizer codes, the tuples $L_X,L_Z$ describe the logical operators that are used to encode message qudits, while elements of $Q_X\cap\spn\{L_Z\}^\perp$ and $Q_Z\cap\spn\{L_X\}^\perp$ correspond to the gauge operators. We remark that some works (e.g.~\cite{bacon_sparse_2017}) have studied subsystem codes in which these spaces of gauge operators have low-weight generators. However, such codes are not neccessarily LDPC in the sense we consider, in which the stabilizer spaces $Q_Z^\perp$ and $Q_X^\perp$ must have low-weight generators.
We now introduce some (slightly nonstandard) notation to describe subsystem codes in the language of chain complexes.

\begin{definition}
  \label{def:subsystemcomplex}
  For $i\in\bZ$, an \textbf{$i$-subsystem chain complex} $(\cC,L)$ consists of a chain complex $\cC$ together with a pair $L=(L_i,L^i)$ of tuples $L^i=(c^1\,\dots,c^k)\in Z^i(\cC)^k$, $L_i=(c_1,\dots,c_k)\in Z_i(\cC)^k$ of the same length $k$, such that $c^j\cdot c_\ell=\1_{j=\ell}$ for all $j,\ell\in[k]$. The \textbf{associated quantum subsystem code} $(Q,L)$ has $Q_X=Z^i(\cC)$, $Q_Z=Z_i(\cC)$, $L_X=L^i$, $L_Z=L_i$.

  The \textbf{$i$-systolic (subsystem) distance $d_i(\cC,L)$} and the \textbf{$i$-cosystolic (subsystem) distance $d^i(\cC,L)$} are defined as
  \begin{equation*}
    d_i(\cC,L) = \min_{c\in Z_i(\cC)\setminus\spn\{L^i\}^\perp}|c|, \hspace{2em} d^i(\cC,L) = \min_{c\in Z^i(\cC)\setminus\spn\{L_i\}^\perp}|c|.
  \end{equation*}
\end{definition}

By definition, the distance $d(Q,L)$ of the subsystem code $(Q,L)$ associated to level $i$ of an $i$-subsystem chain complex $(\cC,L)$ satisfies
\begin{equation*}
  d(Q,L) = \min\{d_i(\cC,L),d^i(\cC,L)\}.
\end{equation*}
Thus we may consider subsystem codes using chain complexes; note that the tuples $(L^i,L_i)=(L_X,L_Z)$ are in fact partial dual $i$-(co)homology bases, in the sense that the tuples can be extended into complete dual $i$-(co)homology bases.

We also naturally extend the notion of homological products to subsystem chain complexes:

\begin{definition}
  \label{def:subsystemproduct}
  Let $(\cA,L^{\cA})$ be an $i$-subsystem chain complex, and let $(\cB,L^{\cB})$ be a $j$-subsystem chain complex. Then the \textbf{homological product}
  \begin{equation*}
    (\cC,L^{\cC})=(\cA,L^{\cA})\otimes(\cB,L^{\cB})
  \end{equation*}
  is the $(i+j)$-subsystem chain complex defined as follows. We let
  \begin{equation*}
    \cC := \cA\otimes\cB
  \end{equation*}
  be the homological product of chain complexes. If $L^{\cA}=((a^1,\dots,a^{k^{\cA}}),(a_1,\dots,a_{k^{\cA}}))$ and $L^{\cB}=((b^1,\dots,b^{k^{\cB}}),(b_1,\dots,b_{k^{\cB}}))$, then we let
  \begin{equation*}
    L^{\cC} := \left((a^{\ell}\otimes b^m)_{\ell\in[k^{\cA}],m\in[k^{\cB}]},\;(a_{\ell}\otimes b_m)_{\ell\in[k^{\cA}],m\in[k^{\cB}]}\right).
  \end{equation*}
\end{definition}

That is, in the language of stabilizer codes, the subsystem homological product simply tensors the chosen bases of logical operators.

\begin{remark}
  \label{remark:subsystemprod}
  Even if $L^{\cA}$ and $L^{\cB}$ consist of complete pairs of dual $i$- and $j$-cohomology/homology bases respectively, the resulting $L^{\cC}$ from Definition~\ref{def:subsystemproduct} will not in general be a complete pair of dual $(i+j)$-cohomology/homology bases. Specifically, the K\"{u}nneth formula implies that 
  \begin{equation}
    \label{eq:Hipjkun}
    H_{i+j}(\cC) \cong \bigoplus_{k\in\bZ}H_k(\cA)\otimes H_{i+j-k}(\cB)
  \end{equation}
  with a similar expression also holding for the cohomology $H^{i+j}(\cC)$. Every homology representative in $L^{\cC}$ by definition has the form $a_\ell\otimes b_m$ for representatives $a_\ell,b_m$ of $H_i(\cA),H_j(\cB)$ respectively, which implies that $a_\ell\otimes b_m$ is a representative of $H_i(\cA)\otimes H_j(\cB)$. Thus every homology representative in $L^{\cC}$ lies in the $k=i+j$ term of the direct sum in~(\ref{eq:Hipjkun}); the same conclusion holds for the cohomology representatives. Therefore $L^{\cC}$ cannot contain complete $(i+j)$-cohomology/homology bases if some $k\neq i+j$ has $H_k(\cA)\otimes H_{i+j-k}(\cB)\neq 0$.
\end{remark}

\subsection{Single-Sector Chain Complexes}
\label{sec:sscc}
While the presentation of chain complexes in Section~\ref{sec:cc} is fairly standard, sometimes it is helpful to restrict attention to a specific class of chain complexes $\cC$ in which all vector spaces $C_i$ for $i\in\bZ$ are the same, and all boundary maps $\partial_i$ are also the same. To the best of our knowledge, the study of such complexes in the context of homological product codes was initiated in \cite{bravyi_homological_2014}, who termed such objets ``single-sector chain complexes.'' We will therefore sometimes refer to ordinary chain complexes, such as those considered in Section~\ref{sec:cc}, as ``multi-sector chain complexes.''

\begin{definition}
  \label{def:sscc}
  A \textbf{single-sector chain complex $\cC_*=(C,\partial^{\cC})$ over $\bF_q$} consists of a $\bF_q$-vector space $C$ and a linear \textbf{boundary map} $\partial^{\cC}:C\rightarrow C$ satisfying $(\partial^{\cC})^2=0$. When clear from context, we omit the superscript and write $\partial=\partial^{\cC}$. The associated \textbf{single-sector cochain complex $\cC^*$} has vector space $C$ and \textbf{coboundary map} $\delta=\partial^\top$. The \textbf{quantum CSS code $Q=(Q_X,Q_Z)$ associated to $\cC_*$} is given by $Q_X=\ker\delta$ and $Q_Z=\ker\partial$.

  Interpreting $\cC_*$ as an ordinary (multi-sector) chain complex with infinitely many equal terms
  \begin{equation*}
    \cdots\xrightarrow{\partial}C\xrightarrow{\partial}C\xrightarrow{\partial}\cdots,
  \end{equation*}
  we define the variables
  \begin{equation*}
    Z_*(\cC),Z^*(\cC),B_*(\cC),B^*(\cC),H_*(\cC),H^*(\cC),d_*(\cC),d^*(\cC),\mu_*(\cC),\mu^*(\cC),M_*(\cC),M^*(\cC)
  \end{equation*}
  exactly as in Definition~\ref{def:cc} and Definition~\ref{def:sysdisexp}, with any $i\in\bZ$ replacing $*$ (the choice of $i$ does not matter by symmetry).
\end{definition}

While Definition~\ref{def:sscc} treats single-sector chain complexes as infinite multi-sector chain complexes with all terms equal, we follow \cite{bravyi_homological_2014} in using a notion of homological product for single-sector complexes, presented below, that is \textit{not} equivalent to taking the multi-sector homological product of these infinite multi-sector complexes. We will restrict attention to fields of characteristic $2$ when considering single-sector products in order to improve simplicity in the presentation by avoiding signing issues.

\begin{definition}
  \label{def:sshomprod}
  For single-sector chain complexes $\cA$ and $\cB$ over a field $\bF_q$ of characteristic $2$, the \textbf{(single-sector) homological product} $\cC=\cA\otimes\cB$ is the single-sector chain complex over $\bF_q$ given by the vector space
  \begin{equation*}
    C := A \otimes B
  \end{equation*}
  and the boundary map
  \begin{equation*}
    \partial^{\cC} := \partial^{\cA}\otimes I+I\otimes\partial^{\cB}.
  \end{equation*}
\end{definition}

Given $t$ single-sector complexes $(\cC_i=(C_i,\partial_i))_{i\in[t]}$, we may take $t-1$ products as defined in Definition~\ref{def:sshomprod} to obtain a single-sector complex $\cA=\cC_1\otimes\cdots\otimes\cC_t$ given by the vector space
\begin{equation*}
  A = C_1\otimes\cdots\otimes C_t
\end{equation*}
with boundary map
\begin{equation*}
  \partial^{\cA} = \partial_1\otimes I^{\otimes t-1} + I\otimes\partial_2\otimes I^{\otimes t-2} + \cdots + I^{\otimes t-1}\otimes\partial_t.
\end{equation*}
Such higher-order products are considered in Theorem~\ref{thm:sspe} in Section~\ref{sec:peprod}.

The K\"{u}nneth formula also holds for single-sector chain complexes:

\begin{proposition}[K\"{u}nneth formula for single-sector complexes; see \cite{bravyi_homological_2014}]
  \label{prop:sskunneth}
  Let $\cA$ and $\cB$ be single-sector chain complexes over a field of characteristic $2$. Then
  \begin{equation*}
    H_*(\cA\otimes\cB) \cong H_*(\cA)\otimes H_*(\cB).
  \end{equation*}
  Furthermore, for $a\in Z_*(\cA)$ and $b\in Z_*(\cB)$, the isomorphism above maps
  \begin{equation*}
    a\otimes b+B_*(\cA\otimes\cB) \mapsfrom (a+B_*(\cA))\otimes(b+B_*(\cB)).
  \end{equation*}
\end{proposition}

Below we present some useful corollaries of the K\"{u}nneth formula.

\begin{corollary}
  \label{cor:ssprodbases}
  If $\{a^1,\dots,a^{k^{\cA}}\},\{a_1,\dots,a_{k^{\cA}}\}$ and $\{b^1,\dots,b^{k^{\cB}}\},\{b_1,\dots,b_{k^{\cB}}\}$ are dual cohomology/homology bases for single-sector chain complexes $\cA$ and $\cB$ respectively, then
  \begin{equation*}
    \{a^i\otimes b^j:i\in[k^{\cA}],j\in[[k^{\cB}]\},\{a_i\otimes b_j:i\in[k^{\cA}],j\in[[k^{\cB}]\}
  \end{equation*}
  are dual cohomology/homology bases for $\cA\otimes\cB$.
\end{corollary}

\begin{corollary}
  \label{cor:ssprodcycles}
  For single-sector chain complexes $\cA_1,\dots,\cA_t$,
  \begin{equation*}
    Z_*(\cA_1\otimes\cdots\otimes\cA_t) = Z_*(\cA_1)\otimes\cdots\otimes Z_*(\cA_t) + B_*(\cA_1\otimes\cdots\otimes\cA_t).
  \end{equation*}
\end{corollary}

Corollary~\ref{cor:ssprodbases} and Corollary~\ref{cor:ssprodcycles} have multi-sector analogues as well, though in this paper we only need the single-sector versions of these statements.


\subsection{Product-Expansion}
\label{sec:pe}
In this section, we review the notion of product-expansion; we generally follow the exposition in \cite{kalachev_two-sided_2023}, which introduced the general high-dimensional formulation of this notion; the two-dimensional version was used to construct asymptotically good qLDPC codes in \cite{panteleev_asymptotically_2022,leverrier_quantum_2022-1,dinur_good_2023}. Product-expansion is a property of a finite set of classical codes, which generalizes ordinary code distance to such larger collections of codes. At a high level, it describes the extent to which elements of the dual tensor (see Definition~\ref{def:classtensor}) of these codes have natural low-weight decompositions.

To define product-expansion, we will need the following notation:

\begin{definition}
  \label{def:dtnot}
  Fix some $t\in\bN$ and some classical codes $C_i\subseteq\bF_q^{n_i}$ for $i\in[t]$. Then for every $i\in[t]$, define $C^{(i)}\subseteq\bigotimes_{j\in[t]}\bF_q^{n_j}=\bF_q^{\prod_{j\in[t]}n_j}$ by
  \begin{equation*}
    C^{(i)}=\left(\bigotimes_{j=1}^{i-1}\bF_q^{n_j}\right)\otimes C_i\otimes\left(\bigotimes_{j=i+1}^{t}\bF_q^{n_j}\right),
  \end{equation*}
  and for every $i,j\in[t]$, define
  \begin{equation*}
    C^{(i,j)} := C^{(i)}\cap C^{(j)}.
  \end{equation*}
  When $t$ is not clear from context, we write $C^{(i;t)}=C^{(i)}$ and $C^{(i,j;t)}=C^{(i,j)}$.
\end{definition}

In words, $C^{(i)}$ is the space of $t$-dimensional tensors for which every direction-$i$ column (in which the $i$th coordinate varies while the other $t-1$ coordinates are fixed) lies in $C_i$. Similarly, for $i\neq j$, then $C^{(i,j)}$ is the space of $t$-dimensional tensors for which every direction-$(i,j)$ plane (in which the $i$th and $j$th coordinates vary while the other $t-2$ coordinates are fixed) lies in $C_i\otimes C_j$.

Definition~\ref{def:classtensor} and Definition~\ref{def:dtnot} immediately imply that
\begin{equation}
  \label{eq:boxplus}
  C_1\boxplus\cdots\boxplus C_t = C^{(1)}+\cdots+C^{(t)}.
\end{equation}
Product-expansion measures the extent to which the sum on the right hand side of~(\ref{eq:boxplus}) can cancel to yield a low-weight element of $C_1\boxplus\cdots\boxplus C_t$ from a sum of high-weight elements of $C^{(1)},\dots,C^{(t)}$. Of course, some cancellations can occur simply because $C^{(i,j)}=C^{(i)}\cap C^{(j)}$ is nonzero for all $i,j$, assuming that $C_i,C_j$ are nonzero. Product-expansion is therefore defined to measure the extent to which ``non-trivial'' cancellations can occur in~(\ref{eq:boxplus}), which are not captured by the spaces $C^{(i,j)}$, as formalized below.


We will need the following modified notion of Hamming weight for elements of $C^{(i)}$.

\begin{definition}
  \label{def:dirham}
  Defining $C^{(i)}$ as in Definition~\ref{def:dtnot}, then for $c\in C^{(i)}$, we let $|c|_i$ denote the number of nonzero direction-$i$ columns in $c$. Formally,
  \begin{equation*}
    |c|_i = \left|\left\{(k_1,\dots,k_{i-1},k_{i+1},\dots,k_t)\in\prod_{j\in[t]\setminus\{i\}}[n_j]:\exists k_i\in[n_i]\text{ with }c_{(k_1,\dots,k_t)}\neq 0\right\}\right|.
  \end{equation*}
\end{definition}

We are now ready to define product-expansion.

\begin{definition}[\cite{kalachev_two-sided_2023}]
  \label{def:pe}
  Fix some $t\in\bN$ and some prime power $q$. The \textbf{product-expansion $\rho$} of a collection $(C_i\subseteq\bF_q^{n_i})_{i\in[t]}$ of classical codes is the largest real number $\rho\geq 0$ such that for every $c\in C_1\boxplus\cdots\boxplus C_t$, there exists a decomposition
  \begin{equation*}
    c=c_1+\cdots+c_t
  \end{equation*}
  with each $c_i\in C^{(i)}$ such that
  \begin{equation*}
    |c| \geq \rho\sum_{i\in[t]}n_i|c_i|_i.
  \end{equation*}
\end{definition}

For some basic intuition, observe that when $t=1$, a single code $C\subseteq\bF_q^n$ of distance $d$ has product-expansion $\rho=d/n$ equal to the relative distance of $C$.

Sometimes we have a given decomposition $c=c_1+\cdots+c_t$ with each $c_i\in C^{(i)}$, and we wish to transform it into a low-weight decomposition $c'=c_1'+\cdots+c_t'$ guaranteed to exist by product-expansion. The following lemma shows that any two decompositions of $c$ differ by sums of elements of the $C^{(i,j)}$'s.

\begin{lemma}[Follows from \cite{kalachev_two-sided_2023}]
  \label{lem:homvanexp}
  Let $c\in C_1\boxplus\cdots\boxplus C_t$. For any two decompositions
  \begin{equation*}
    c=c_1+\cdots+c_t=c_1'+\cdots+c_t'
  \end{equation*}
  with $c_i,c_i'\in C^{(i)}$ for all $i\in[t]$, then there exist choices of $c_{i,j}\in C^{(i,j)}$ for $1\leq i<j\leq t$ such that
  \begin{equation}
    \label{eq:homvanexp}
    c_i-c_i' = \sum_{j=1}^{i-1}c_{j,i}-\sum_{j=i+1}^tc_{i,j}
  \end{equation}
  for all $i\in[t]$.
\end{lemma}

Lemma~\ref{lem:homvanexp} essentially follows from the discussion in Appendix~B of \cite{kalachev_two-sided_2023}, though for completeness we provide a proof in Appendix~\ref{sec:peproofs}.

The following lemma shows that passing to subcodes preserves product-expansion, up to some loss.

\begin{lemma}
  \label{lem:pesubmain}
  Let $\bF_q$ be a field of characteristic $2$. If a collection $(C_i\subsetneq\bF_q^n)_{i\in[t]}$ of classical codes has product-expansion at least $\rho>0$, then there exists some $\rho'=\rho'(\rho,t)>0$ depending only on $\rho$ and $t$ such that every collection $(C_i'\subsetneq\bF_q^n)_{i\in[t]}$ of codes with each $C_i'\subseteq C_i$ has product-expansion at least $\rho'$.
\end{lemma}

As mentioned above, product-expansion can be viewed as a high-dimensional generalization of classical code distance, which can only increase when passing to subcodes. In this lens, Lemma~\ref{lem:pesubmain} may not be too surprising. Nevertheless, our proof of Lemma~\ref{lem:pesubmain}, provided in Appendix~\ref{sec:peproofs}, is not entirely straightforward.

\begin{remark}
  \label{remark:dlvrestrict}
  In Lemma~\ref{lem:pesubmain}, we assume that $\bF_q$ has characteristic $2$ and that all codes $C_i$ have the same length $n$ simply because our proof (see Appendix~\ref{sec:peproofs}) relies on a result of \cite{dinur_expansion_2024} (see Lemma~\ref{lem:petocpe}), where a proof was only given with these restrictions. We believe all the arguments should generalize naturally to arbitrary $\bF_q$ and to codes $C_i$ of different lengths $n_i$. However, we will not need such a generalization for the purpose of our paper, so for simplicity we work in the same restricted setting as \cite{dinur_expansion_2024} here.
\end{remark}

The following results on the product-expansion of two random or two Reed-Solomon codes were previously known:

\begin{theorem}[\cite{kalachev_two-sided_2023,dinur_good_2023}]
  \label{thm:pe2rand}
  For every $\epsilon>0$, there exists a real number $\rho=\rho(\epsilon)>0$ such that the following holds: For every $n\in\bN$ and every $k_1,k_2\leq(1-\epsilon)n$, a uniformly random pair of codes $C_1,C_2\subseteq\bF_q^n$ of respective dimensions $k_1,k_2$ has product-expansion $\geq\rho$ with probability approaching $1$ as $n\rightarrow\infty$.
\end{theorem}

\begin{theorem}[\cite{polishchuk_nearly-linear_1994}]
  \label{thm:pe2RS}
  For every $\epsilon>0$, there exists $\rho=\rho(\epsilon)>0$ such that the following holds: For every prime power $q$ and every\footnote{\cite{polishchuk_nearly-linear_1994} only presented their proof for the case where $k_1=k_2$. However, the proof goes through flawlessly when $k_1\neq k_2$, as was for instance observed in \cite{kalachev_two-sided_2023}.} $k_1,k_2\in\bN$ with $k_1+k_2\leq(1-\epsilon)n$, the codes $(\textrm{RS}(q,k_1),\textrm{RS}(q,k_2))$ have product-expansion $\geq\rho$.
\end{theorem}

Kalachev \& Panteleev \cite{kalachev_personal_2024} recently showed the following result on the product-expansion of an arbitrary constant number of random codes, over a sufficiently large alphabet:

\begin{theorem}[\cite{kalachev_personal_2024}]
  \label{thm:petrand}
  For every $t\in\bN$ and every collection of intervals $I_1,\dots,I_t\subseteq(0,1)$, there exists a real number $\rho>0$ such that the following holds: for every $n,k_1,\dots,k_t\in\bN$ such that $k_i/n\in I_i$ for each $i\in[t]$, letting $q=2^{(n+3)^t}$, then a uniformly random tuple of linear codes $C_1,\dots,C_t\subseteq\bF_q^n$ of respective dimensions $k_1,\dots,k_t$ has product-expansion $\geq\rho$ with probability approaching $1$ as $n\rightarrow\infty$.
\end{theorem}


\section{Homological Products of Product-Expanding Codes}
\label{sec:peprod}
In this section, we bound the distance of homological products of single-sector chain complexes based on product-expanding classical codes. We generalize and strengthen the results of \cite{bravyi_homological_2014} on homological products of single-sector chain complexes based on random codes. We also apply these techniques to single-sector chain complexes from Reed-Solomon codes.

The following theorem is the main result of this section.

\begin{theorem}
  \label{thm:sspe}
  Fix some $t\in\bN$ and some field $\bF_q$ of characteristic $2$. For each $i\in[t]$ let $\cC_i=(C_i,\partial_i)$ be a single-sector chain complex over $\bF_q$, and let $n_i=\dim C_i$. For each $i\in[t]$, let $\rho_i$ denote the product-expansion of the collection of codes $(B_*(\cC_1),\dots,B_*(\cC_{i-1}),Z_*(\cC_i))$. Then the homological product $\cA=(A,\partial^{\cA}):=\bigotimes_{i\in[t]}\cC_i$ has systolic distance
  \begin{equation}
    \label{eq:sspedis}
    d_*(\cA) \geq \prod_{i\in[t]}(\rho_in_i).
  \end{equation}
  
  Furthermore, if $t=2$, letting $\rho'$ denote the product-expansion of $(B_*(\cC_1),B_*(\cC_2))$ and letting $\Delta_i=d(Z_*(\cC_i))/n_i$ denote the relative distance of the classical code $Z_*(\cC_i)$, then the homological product $\cA$ has filling constant
  \begin{equation}
    \label{eq:sspeexp}
    \mu_*(\cA) \leq \frac{1}{\rho'\cdot\min\{\rho',\Delta_1,\Delta_2\}}.
  \end{equation}
\end{theorem}

We first describe some applications of Theorem~\ref{thm:sspe}, before proving the theorem in Section~\ref{sec:sspeproof}.

To begin, we present a simple method of turning a CSS code into a single-sector chain complex; a similar method was used in \cite{bravyi_homological_2014}.

\begin{lemma}
  \label{lem:codetocomplex}
  Let $Q=(Q_X,Q_Z)$ be a CSS code of length $n$ over $\bF_q$ of characteristic $2$ such that $\dim Q_X=\dim Q_Z$. Fix some full-rank parity check matrices $H_X,H_Z$, so that $Q_X=\ker H_X$, $Q_Z=\ker H_Z$. Define a single-sector chain complex $\cC$ with vector space $C=\bF_q^n$ and $\partial=H_X^\top H_Z$. Then $Q$ equals the quantum code associated to $\cC$.
\end{lemma}
\begin{proof}
  Since $H_ZH_X^\top=0$ by the CSS orthogonality condition, we have $\partial^2=H_X^\top H_ZH_X^\top H_Z=0$. Therefore $\cC$ is a well-defined single-sector chain complex. Now because $H_X,H_Z$ are full-rank, $\ker\delta=\ker H_X=Q_X$ and $\ker\partial=\ker H_Z=Q_Z$, as desired.
\end{proof}

\subsection{Product of Two Codes}
\cite{bravyi_homological_2014} showed that the homological product of $t=2$ two (appropriately sampled) random single-sector chain complexes yields asymptotically good $[[N=n^2,\Theta(N),\Theta(N)]]$ quantum codes of locality $w=\Theta(\sqrt{N})$. Theorem~\ref{thm:sspe} strengthens this result in a couple of ways. First, the following corollary shows that such random product codes are also locally testable:

\begin{corollary}
  \label{cor:ssperandom}
  For every $\epsilon>0$, there exist $\Delta=\Delta(\epsilon)>0$ and $\rho=\rho(\epsilon)>0$ such that the following holds: For $i\in[2]$, let $Q^i=(Q_X^i,Q_Z^i)$ be a uniformly random length-$n$, rate-$R_i$ CSS code over any $\bF_q$ of characteristic $2$ such that $\dim(Q_X^i)=\dim(Q_Z^i)$, and such that the chosen rates satisfy $R_1,R_2\leq 1-\epsilon$. Let $\cC_i$ be the single-sector complex obtained from $Q^i$ via Lemma~\ref{lem:codetocomplex}. Then with probability approaching $1$ as $n\rightarrow\infty$, the quantum code associated to the homological product $\cA=\cC_1\otimes\cC_2$ is a $[[n^2,\;R_1R_2\cdot n^2,\;\Delta\cdot n^2]]$ code that is locally testable with locality $\leq 2n$ and soundness $\geq\rho$.
\end{corollary}
\begin{proof}
  The corollary follows immediately from Theorem~\ref{thm:sspe}, Lemma~\ref{lem:codetocomplex}, and Theorem~\ref{thm:pe2rand}, as in a uniformly random CSS code $Q^i$, all the codes $Q_X^i,Q_Z^i,{Q_X^i}^\perp,{Q_Z^i}^\perp$ are uniformly random classical codes (of the specified dimensions).
\end{proof}

Below, we also show that for appropriate choices of the rates of $Q^1,Q^2$, the random CSS codes above can be replaced by quantum Reed-Solomon codes, thereby yielding an explicit construction of good length-$N$ quantum codes with locality $O(\sqrt{N})$ and constant soundness. The key idea is to apply the product-expansion of Reed-Solomon codes shown by \cite{polishchuk_nearly-linear_1994} (see Theorem~\ref{thm:pe2RS}).

\begin{corollary}
  \label{cor:sspeRS}
  For every $\epsilon>0$, there exist $\Delta=\Delta(\epsilon)>0$ and $\rho=\rho(\epsilon)>0$ such that the following holds: For $i\in[2]$, let $Q^i=(Q_X^i,Q_Z^i)$ be a length-$n$, rate-$R_i$ quantum Reed-Solomon code over some $\bF_q$ of characteristic $2$ such that $\dim(Q_X^i)=\dim(Q_Z^i)$, and such that the chosen rates satisfy $R_1,R_2\leq 1-\epsilon$ and $|R_1-R_2|\geq\epsilon$. Let $\cC_i$ be the single-sector complex obtained from $Q^i$ via Lemma~\ref{lem:codetocomplex}. Then the quantum code associated to the homological product $\cA=\cC_1\otimes\cC_2$ is a $[[n^2,\;R_1R_2\cdot n^2,\;\Delta\cdot n^2]]$ code that is locally testable with locality $\leq 2n$ and soundness $\geq\rho$.
\end{corollary}
\begin{proof}
  The corollary follows immediately from Theorem~\ref{thm:sspe}, Lemma~\ref{lem:codetocomplex}, and Theorem~\ref{thm:pe2RS}. Specifically, assuming without loss of generality that $R_1-R_2>\epsilon$, then $B_*(\cC_1),\;Z_*(\cC_2)$ are classical Reed-Solomon codes of rates $(1-R_1)/2,\;(1+R_2)/2$ respectively, so Theorem~\ref{thm:pe2RS} ensures that this pair of codes has product-expansion at least some $\rho=\rho(\epsilon/2)>0$. We may similarly conclude that $(B_*(\cC_1),B_*(\cC_2))$ has product-expansion at least $\rho$, so we may apply Theorem~\ref{thm:sspe} to conclude the desired result.
\end{proof}

While product-expanding random codes has been used for the ``local codes'' in constructions of asymptotically good quantum LDPC codes \cite{panteleev_asymptotically_2022,leverrier_quantum_2022-1,dinur_good_2023}, the product-expansion of Reed-Solomon codes seems to be too weak for this purpose, as observed in \cite{kalachev_two-sided_2023}. Specifically, this application requires classical codes $C_1,C_2$ such that both $(C_1,C_2)$ and $(C_1^\perp,C_2^\perp)$ are product-expanding, which cannot be obtained using Theorem~\ref{thm:pe2RS} due to the requirement that $R_1+R_2<1$. \cite{kalachev_two-sided_2023} explains how this issue seems to stem from the fact that the dual of a Reed-Solomon code is another Reed-Solomon code, but $(C_1,C_2)$ will have poor product-expansion if $C_1^\perp\subseteq C_2$. It is therefore perhaps surprising that we are able to use the product-expansion of quantum Reed-Solomon codes in Corollary~\ref{cor:sspeRS} to obtain good quantum codes with nontrivial locality.

\subsection{Higher Order Products}
Theorem~\ref{thm:sspe} also provides a way of obtaining higher-dimensional homological products with good distance, if we are given CSS codes satisfying appropriate product-expansion properties. \cite{bravyi_homological_2014} conjectured that such higher products with good distance would exist, so our result is a conditional affirmation of their conjecture.




The following corollary resolves this conjecture of \cite{bravyi_homological_2014} in the affirmative over sufficiently large alphabets. Specifically, applying the product-expansion result of Kalachev \& Panteleev described in Theorem~\ref{thm:petrand} with Theorem~\ref{thm:sspe} immediately yields the following:

\begin{corollary}
  \label{cor:sspemanyrandom}
  For every $\epsilon>0$ and $t\in\bN$, there exist $\Delta=\Delta(\epsilon,t)>0$ such that the following holds: For $i\in[t]$, let $Q^i=(Q_X^i,Q_Z^i)$ be a uniformly random length-$n$, rate-$R_i$ CSS code over $\bF_q$, $q=2^{(n+3)^t}$, such that $\dim(Q_X^i)=\dim(Q_Z^i)$, and such that the chosen rates satisfy $R_i\leq 1-\epsilon$. Let $\cC_i$ be the single-sector complex obtained from $Q^i$ via Lemma~\ref{lem:codetocomplex}. Then with probability approaching $1$ as $n\rightarrow\infty$, the quantum code associated to the homological product $\cA=\cC_1\otimes\cdots\otimes\cC_t$ is a $[[n^t,\; (\prod_{i\in[t]}R_i)\cdot n^t,\; \Delta\cdot n^t]]_q$ code with locality $\leq tn$.
\end{corollary}
\begin{proof}
  By the definition of a random CSS code, $(Q_X^i)_{i\in[t]}$ and $(Q_Z^i)_{i\in[t]}$ are both collections of random linear codes in $\bF_q^n$, subject to the $i$'th code in each collection having rate $(1+R_i)/2$; by definition the codes within each collection are independently drawn, while there are dependencies between the collections. Therefore by Theorem~\ref{thm:petrand}, there exists some $\rho=\rho(\epsilon,t)>0$ such that both $(Q_X^i)_{i\in[t]}$ and $(Q_Z^i)_{i\in[t]}$ are $\rho$-product-expanding with probability $\rightarrow 1$ as $n\rightarrow\infty$. Then Theorem~\ref{thm:sspe} along with Lemma~\ref{lem:pesubmain} immediately implies the desired result.
\end{proof}

While Corollary~\ref{cor:sspemanyrandom} provides asymptotically good $[[N=n^t,\Theta(N),\Theta(N)]]_q$ codes, the alphabet size $q=2^{(n+3)^t}$ grows exponentially in the block length, and the locality $tn=tN^{1/t}$ grows as a small polynomial in the block length for large constant $t$. Reducing the alphabet size would require proving product-expansion of collections of $t$ codes over some smaller alphabet than achieved in Theorem~\ref{thm:petrand}, which remains an interesting open question. The locality could potentially be reduced using a locality-reduction procedure like that of \cite{hastings_quantum_2023} (see Theorem~\ref{thm:locred}). However, \cite{hastings_quantum_2023} only considers binary alphabets, so their techniques would need to be generalized to larger alphabets in order to apply to the codes in Corollary~\ref{cor:sspemanyrandom}. 

\subsection{Proof of Theorem~\ref{thm:sspe}}
\label{sec:sspeproof}
We prove the two bounds in Theorem~\ref{thm:sspe} separately. The first bound~(\ref{eq:sspedis}) will follow directly from the following lemma. Below, we define all variables as in Theorem~\ref{thm:sspe}. Similarly as in Section~\ref{sec:pe}, for a subspace $V\subseteq C_i$, we let $V^{(i)}\subseteq A=\bigotimes_{i\in[t]}C_i$ denote the space of $t$-dimensional tensors in which every direction-$i$ column lies in $V$. For a linear map $f$ acting on $C_i$, we similarly let $f^{(i)}=I^{\otimes i-1}\otimes f\otimes I^{t-i}$ denote the map that simply applies $f$ to each direction-$i$ column. When $t$ is not clear from context, we write $V^{(i;t)}=V^{(i)}$ and $f^{(i;t)}=f^{(i)}$.

\begin{lemma}
  \label{lem:sspedis}
  Fix any $t\in\bN$ and define all variables as in Theorem~\ref{thm:sspe}. Then for every
  \begin{equation*}
    a \in \bigotimes_{i\in[t]}Z_*(\cC_i)+\sum_{i\in[t]}B_*(\cC_i)^{(i)}
  \end{equation*}
  such that there exist $z^i\in Z^*(\cC_i)$ for $i\in[t]$ satisfying
  \begin{equation*}
    \left(\bigotimes_{i\in[t]}z^i\right)\cdot a \neq 0,
  \end{equation*}
  it holds that
  \begin{equation}
    \label{eq:pedisbound}
    |a| \geq \prod_{i\in[t]}(\rho_in_i).
  \end{equation}
\end{lemma}
\begin{proof}
  We prove the lemma by induction on $t$. For the base case, if $t=1$, then $a\in Z_*(\cC_1)$ and $z^1\in Z^*(\cC_1)$ with $a\cdot z^1\neq 0$. It follows that $a$ and $z^1$ are representatives of nontrivial homology and cohomology classes respectively, so $|a|\geq d_*(\cC_1)\geq\rho_1n_1$, and~(\ref{eq:pedisbound}) holds, as desired.

  For the inductive step, let $t\geq 2$, and assume that the lemma holds for $t-1$. Choose any $a_0\in\bigotimes_{i\in[t]}Z_*(\cC_i)$ and $a_i\in B_*(\cC_i)^{(i)}$ for $i\in[t]$ such that $a=a_0+a_1+\cdots+a_t\in A$. By definition $a_0$ and $a_t$ both lie in $Z_*(\cC_t)^{(t)}$, so
  \begin{equation*}
    a_1+\cdots+a_{t-1}+(a_0+a_t) \in B_*(\cC_1)^{(1)}+\cdots+B_*(\cC_{t-1})^{(t-1)}+Z_*(\cC_t)^{(t)}.
  \end{equation*}
  Recalling that $\rho_t$ denotes the product-expansion of $(B_*(\cC_1),\dots,B_*(\cC_{t-1}),Z_*(\cC_t))$, Definition~\ref{def:pe} with Lemma~\ref{lem:homvanexp} implies that there exist choices of $a_{i,t}\in B_*(\cC_i)^{(i)}\cap Z_*(\cC_t)^{(t)}$ for all $i\in[t-1]$ such that
  \begin{equation}
    \label{eq:applype}
    |a| \geq \rho_tn_t\biggl|a_0+a_t+\sum_{i\in[t-1]}a_{i,t}\biggr|_t,
  \end{equation}
  where as in Definition~\ref{def:dirham}, $|\cdot|_t$ denotes the number of nonzero direction-$t$ columns in a given tensor.

  Define $a'\in\bigotimes_{i\in[t-1]}C_i$ by
  \begin{equation*}
    a' := (I^{\otimes t-1}\otimes {z^t}^\top)\biggl(a_0+a_t+\sum_{i\in[t-1]}a_{i,t}\biggr) = (I^{\otimes t-1}\otimes {z^t}^\top)\biggl(a_0+\sum_{i\in[t-1]}a_{i,t}\biggr),
  \end{equation*}
  where the second equality above holds because every direction-$t$ column of $a_t$ by definition lies in $B_*(\cC_t)$ and is therefore orthogonal to $z^t\in Z^*(\cC_t)$. Thus we find that $a'=a_0'+\sum_{i\in[t-1]}a_i'$ for vectors
  \begin{align*}
    a_0' &:= (I^{\otimes t-1}\otimes{z^t}^\top)a_0 \in \bigotimes_{i\in[t-1]}Z_*(\cC_i) \\
    a_i' &:= (I^{\otimes t-1}\otimes{z^t}^\top)a_{i,t} \in B_*(\cC_i)^{(i;t-1)} \hspace{1em} \forall i\in[t-1],
  \end{align*}
  Furthermore,
  \begin{equation*}
    \left(\bigotimes_{i\in[t-1]}z^i\right)\cdot a' = \left(\bigotimes_{i\in[t]}{z^i}^\top\right)a_0 = \left(\bigotimes_{i\in[t]}z^i\right)\cdot a \neq 0,
  \end{equation*}
  as every direction-$i$ column of $a_i$ and $a_i'$ lies in $B_*(\cC_i)$ and therefore is orthogonal to $z^i\in Z^*(\cC_i)$. Note that here we use the fact that dotting a tensor with $\bigotimes_iz^i$ is equivalent to applying the functional $\bigotimes_i{z^i}^\top$ to the tensor.

  Thus we have shown that $a'$ and $z^1,\dots,z^{t-1}$ satisfy the conditions of the lemma statement with $t$ replaced by $t-1$, so the inductive hypothesis implies that
  \begin{equation}
    \label{eq:peapplyind}
    |a'| \geq \prod_{i\in[t-1]}(\rho_in_i).
  \end{equation}
  By definition, a given component of $a'$ can only be nonzero if the respective direction-$t$ column of $a_0+a_t+\sum_{i\in[t-1]}a_{i,t}$ is nonzero. Thus~(\ref{eq:applype}) and~(\ref{eq:peapplyind}) imply that
  \begin{equation*}
    |a| \geq \prod_{i\in[t]}(\rho_in_i),
  \end{equation*}
  completing the inductive step.
\end{proof}

We show below how Lemma~\ref{lem:sspedis} implies the bound~(\ref{eq:sspedis}) in Theorem~\ref{thm:sspe}.

\begin{proof}[Proof of~(\ref{eq:sspedis}) in Theorem~\ref{thm:sspe}]
  For any $a\in Z_*(\cA)\setminus B_*(\cA)$, the K\"{u}nneth formula (Proposition~\ref{prop:sskunneth}, Corollary~\ref{cor:ssprodbases}, and Corollary~\ref{cor:ssprodcycles}) guarantees that
  \begin{equation*}
    a \in \bigotimes_{i\in[t]}Z_*(\cC_i)+\im(\partial^{\cA}) \subseteq \bigotimes_{i\in[t]}Z_*(\cC_i)+\sum_{i\in[t]}B_*(\cC_i)^{(i)},
  \end{equation*}
  and that there exist some $z^i\in Z^*(\cC_i)$ for $i\in[t]$ such that
  \begin{equation}
    \label{eq:finddualelem}
    \left(\bigotimes_{i\in[t]}z^i\right)\cdot a \neq 0.
  \end{equation}
  Specifically, Corollary~\ref{cor:ssprodbases} implies that the set of all $\bigotimes_{i\in[t]}z^i$ for $z^i\in Z^*(\cC_i)$ span the cohomology $H^*(\cC)$, so if no $z^i$ existed satisfying~(\ref{eq:finddualelem}), then the nondegeneracy of the natural cohomology/homology bilinear form would imply that $a$ necessarily lies in $B_*(\cA)$, a contradiction. Thus Lemma~\ref{lem:sspedis} implies that $|a|\geq\prod_{i\in[t]}(\rho_in_i)$, as desired.
\end{proof}

We now prove the second bound (Equation~(\ref{eq:sspeexp})) in Theorem~\ref{thm:sspe}.

\begin{proof}[Proof of~(\ref{eq:sspeexp}) in Theorem~\ref{thm:sspe}]
  Consider any boundary $b\in B_*(\cA)$. Letting $\mu$ denote the right-hand side of~(\ref{eq:sspeexp}), then our goal is to construct a filling of $b$ of weight $\leq\mu|b|$. To begin, choose an arbitrary filling $a\in A$ of $b$, so that
  \begin{equation*}
    \partial^{\cA}a = \partial_1^{(1)}a+\partial_2^{(2)}a = b.
  \end{equation*}
  Applying the $\rho'$-product expansion of $(B_*(\cC_1),B_*(\cC_2))$ to $b$, we conclude by Lemma~\ref{lem:homvanexp} that there exists some $e\in B_*(\cC_1)\otimes B_*(\cC_2)$ such that
  \begin{equation*}
    |b| \geq \rho' \cdot \bigl( n_1\cdot |e+\partial_1^{(1)}a|_1+n_2 \cdot |e+\partial_2^{(2)}a|_2 \bigr) \ .
  \end{equation*}
  Because $e\in\im(\partial_1)\otimes\im(\partial_2)$, there exist preimages $e_1,e_2\in A$ of $e$ under $\partial_1^{(1)},\partial_2^{(2)}$ respectively such that $e_1\in\bF_q^{n_1}\otimes\im(\partial_2)$ and $e_2\in\im(\partial_1)\otimes\bF_q^{n_2}$. Define $a':=a+e_1+e_2$. Then $\partial^{\cA}a'=\partial^{\cA}a+e+e=b$, and
  \begin{equation}
    \label{eq:addes}
    n_1|\partial_1^{(1)}a'|_1+n_2|\partial_2^{(2)}a'|_2 = n_1\cdot |e+\partial_1^{(1)}a|_1+n_2 \cdot |e+\partial_2^{(2)}a|_2 \leq \frac{|b|}{\rho'}.
  \end{equation}

  If $|b|/n_1n_2\geq\rho'\cdot\min\{\Delta_1,\Delta_2\}$, then $|a|\leq n_1n_2\leq\mu|b|$, so $a$ is the desired filling of $b$, and we are done. Therefore assume that $|b|/n_1n_2<\rho'\cdot\min\{\Delta_1,\Delta_2\}$. It follows that
  \begin{align*}
    \frac{|b|}{\rho'n_1} &<\Delta_2n_2 = d(Z_*(\cC_2)) \\
    \frac{|b|}{\rho'n_2} &<\Delta_1n_1 = d(Z_*(\cC_1)).
  \end{align*}
  By~(\ref{eq:addes}), $|\partial_1^{(1)}a'|_1\leq|b|/\rho'n_1$ and $|\partial_2^{(2)}a'|_2\leq|b|/\rho'n_2$. Now for $i=1,2$, let $S_i\subseteq[n_{3-i}]$ denote the set of direction-$i$ vectors in $a'$ that lie outside of $Z_*(\cC_i)$, so that $|S_i|=|\partial_i^{(i)}a'|_i$. Then we have shown that $a'\in A=\bF_q^{n_1\times n_2}$ is a matrix in which $|S_1|<d(Z_*(\cC_2))$ columns lie outside of $Z_*(\cC_1)$, and $|S_2|<d(Z_*(\cC_1))$ rows lie outside of $Z_*(\cC_2)$.

  Therefore there exists some $f\in Z_*(\cC_1)\otimes Z_*(\cC_2)$ that agrees with $a'$ in all columns in $[n_2]\setminus S_1$, and that agrees with $a'$ in all rows in $[n_1]\setminus S_2$. Specifically, we may obtain $f$ by first writing $a'|_{([n_1]\setminus S_2)\times([n_2]\setminus S_1)}\in C_1|_{[n_1]\setminus S_2}\otimes C_2|_{[n_2]\setminus S_1}$ as a sum of pure tensors $\sum_jc_{1,j}\otimes c_{2,j}$, and then letting $f=\sum_j\hat{c}_{1,j}\otimes\hat{c}_{2,j}$, where $\hat{c}_{i,j}$ denotes the unique element of $C_i$ whose restriction to components in $[n_i]\setminus S_{3-i}$ equals $c_{i,j}$.

  Now define $a'':=a'+f$. Then $\partial^{\cA}a''=\partial^{\cA}a'=b$ because $f\in\ker\partial^{\cA}$ by definition. Furthermore,
  \begin{align*}
    |a''|
    &\leq |S_1|\cdot|S_2| = |\partial_1^{(1)}a'|_1\cdot|\partial_2^{(2)}a'|_2 \leq \frac{|b|^2}{{\rho'}^2n_1n_2} \leq \frac{|b|}{{\rho'}^2} \leq \mu|b|,
  \end{align*}
  as desired, where the second inequality above holds by~(\ref{eq:addes}).
\end{proof}


\section{Homological Products of Quantum LDPC and Locally Testable Codes}
\label{sec:subprod}
In this section, we show how to construct quantum LDPC and locally testable codes with large distance as homological products of smaller such codes. Specifically, we apply the techniques of \cite{kaufman_new_2021} to appropriate homological products of subsystem chain complexes. Unlike Section~\ref{sec:peprod}, which considered single-sector chain complexes, this section will consider multi-sector subsystem chain complexes (see Section~\ref{sec:subsystem}).

\subsection{General Bounds on Homological Products}
\label{sec:genbounds}
In this section, we present some general bounds on the systolic distance and filling constants of homological products of multi-term subsystem chain complexes. The proofs here are adaptations of the proof of Theorem~3.1 in \cite{kaufman_new_2021}. However, there is an important conceptual difference: our results apply simultaneously to a chain complex and its cochain complex, due to the symmetric structure of the chain complexes we consider. That is, our results translate to distance and testability bounds for both $X$ and $Z$ errors in quantum codes. In contrast, \cite{kaufman_new_2021} were only able to use these techniques on cochain complexes to bound cosystolic distance and expansion, and needed separate techniques to bound systolic distance.

It is generally difficult to obtain homological products with good (i.e.~close to linear) systolic and cosystolic distances, as demonstrated by the fact that $N$-qudit hypergraph product codes have distance $O(\sqrt{N})$ \cite{tillich_quantum_2014}. In Theorem~\ref{thm:proddis} below, we circumvent this barrier by using subsystem codes. Specifically, the homological products we consider still have low-weight codewords, but these codewords are not used to store logical message qudits.

For simplicity we only present our results for chain complexes of lengths relevant for our purposes. However, similarly as shown in \cite{kaufman_new_2021}, the results below do generalize naturally to chain complexes with arbitrarily many terms.

In the results below, we frequently use the notation and terminology regarding chain complexes described in Section~\ref{sec:cc} and Section~\ref{sec:subsystem}. More motivation and intuition for these results, and specifically Theorem~\ref{thm:proddis}, can be found in Section~\ref{sec:subprodinf}.

\begin{theorem}
  \label{thm:proddis}
  Let $(\cA,L^{\cA})$ be a 0-subsystem chain complex where
  \begin{equation*}
    \cA = (A_1 \xrightarrow{\partial^{\cA}_1} A_0 \xrightarrow{\partial^{\cA}_0} A_{-1})
  \end{equation*}
  has 3 terms, and for some $i\in\bZ$ let $(\cB,L^{\cB})$ be an $i$-subsystem chain complex where
  \begin{equation*}
    \cB = (\cdots B_{i+2} \xrightarrow{\partial^{\cB}_{i+2}} B_{i+1} \xrightarrow{\partial^{\cB}_{i+1}} B_i \xrightarrow{\partial^{\cB}_{i}} B_{i-1} \xrightarrow{\partial^{\cB}_{i-1}} B_{i-2} \cdots)
  \end{equation*}
  has $\geq 5$ terms\footnote{This result, as well as Theorem~\ref{thm:prodfill} below, goes through when $B_{i+2}=0$, so that $\cB$ may have only $4$ terms. However, we include the $B_{i+2}$ term to emphasize the chain/cochain symmetry; all of these results apply equally well to the cochain complexes as to the chain complexes.}. Then the homological product $(\cC,L^{\cC})=(\cA,L^{\cA})\otimes(\cB,L^{\cB})$ has $i$-systolic distance
  \begin{equation*}
    d_i(\cC,L^{\cC}) \geq \frac{d_0(\cA,L^{\cA})\cdot d_i(\cB,L^{\cB})}{w^{\cA}\cdot \mu_i(\cB) \cdot n_{i-1}^{\cB}/d_{i-1}(\cB)} \ ,
  \end{equation*}
  where $n_{i-1}^{\cB}=\dim B_{i-1}$.
\end{theorem}

In Theorem~\ref{thm:iterative} in Section~\ref{sec:iterapp} below, we will repeatedly apply Theorem~\ref{thm:proddis} starting with a constant-sized 5-term chain complex given by \cite{dinur_expansion_2024,kalachev_personal_2024}. This starting complex is not a subsystem complex, but the complex resulting from a single application of Theorem~\ref{thm:proddis} is a subsystem complex (see Remark~\ref{remark:subsystemprod}). Therefore to iteratively apply Theorem~\ref{thm:proddis}, we need to allow $\cA$ to be a subsystem complex. It turns out that Theorem~\ref{thm:proddis} goes through when $\cB$ is a subsystem complex as well, so we present the result in such generality in case it is helpful in future applications.

We now outline the main idea for the proof of Theorem~\ref{thm:proddis}. Without loss of generality, if we re-index the chain complex $\cB$ to have $i=0$, then the homological product $\cC=\cA\otimes\cB$ is a 7-term complex given by the diagram below.

\begin{equation}
  \label{eq:3times5cd}
  \begin{tikzcd}
    A_1\otimes B_2  \arrow[r,"I\otimes\partial_2^{\cB}"] \arrow[d,"\partial_1^{\cA}\otimes I"] & A_1\otimes B_1 \arrow[r,"I\otimes\partial_1^{\cB}"] \arrow[d,"\partial_1^{\cA}\otimes I"] & A_1\otimes B_0 \arrow[r,"I\otimes\partial_{0}^{\cB}"] \arrow[d,"\partial_1^{\cA}\otimes I"] & A_1\otimes B_{-1} \arrow[r,"I\otimes\partial_{-1}^{\cB}"] \arrow[d,"\partial_1^{\cA}\otimes I"] & A_1\otimes B_{-2} \arrow[d,"\partial_1^{\cA}\otimes I"] \\
    A_0\otimes B_2  \arrow[r,"-I\otimes\partial_2^{\cB}"] \arrow[d,"\partial_0^{\cA}\otimes I"] & A_0\otimes B_1 \arrow[r,"-I\otimes\partial_1^{\cB}"] \arrow[d,"\partial_0^{\cA}\otimes I"] & A_0\otimes B_0 \arrow[r,"I\otimes\partial_{0}^{\cB}"] \arrow[d,"\partial_0^{\cA}\otimes I"] & A_0\otimes B_{-1} \arrow[r,"-I\otimes\partial_{-1}^{\cB}"] \arrow[d,"\partial_0^{\cA}\otimes I"] & A_0\otimes B_{-2} \arrow[d,"\partial_0^{\cA}\otimes I"] \\
    A_{-1}\otimes B_2 \arrow[r,"I\otimes\partial_2^{\cB}"] & A_{-1}\otimes B_1 \arrow[r,"I\otimes\partial_1^{\cB}"] & A_{-1}\otimes B_0 \arrow[r,"I\otimes\partial_{0}^{\cB}"] & A_{-1}\otimes B_{-1} \arrow[r,"I\otimes\partial_{-1}^{\cB}"] & A_{-2}\otimes B_0
  \end{tikzcd}
\end{equation}

Specifically, $C_j$ equals the direct sum of the vector spaces $A_k\otimes B_{j-k}$ for $k\in\bZ$, which are simply the vector spaces along the $j$th diagonal in~(\ref{eq:3times5cd}); the boundary map $\partial^{\cC}$ is the sum of the maps in~(\ref{eq:3times5cd}). To prove Theorem~\ref{thm:proddis}, we want to show that every $c\in Z_0(\cC)\setminus\spn\{(L^{\cC})^i\}^\perp$ has large Hamming weight. Write $c=(c_1,c_0,c_{-1})$, where $c_j\in A_j\otimes B_{-j}$.

Letting\footnote{As mentioned in Definition~\ref{def:cc}, for a boundary map $\partial_i$ (such as $\partial^{\cA}_i$ or $\partial^{\cB}_i$), when clear from context we may omit the subscript $i$. In particular, $i$ can often be inferred from the argument to which the boundary map is applied.} $\partial^{(\cA)}=\partial^{\cA}\otimes I$ and $\partial^{(\cB)}=I\otimes\partial^{\cB}$, then the $A_0\otimes B_{-1}$-component of the syndrome $\partial_0^{\cC}c$ by definition equals $\partial^{(\cA)}c_1+\partial^{(\cB)}c_0$. If we had $\partial^{(\cA)}c_1=0$, then every row of the matrix $c_0$ would lie in $Z_0(\cB)$, and we would be able to argue that $|c_0|$ is large, essentially by using the assumption that $c\in\spn\{(L^{\cC})^i\}^\perp$ to relate $c_0$ to a high-weight codeword of the classical tensor code $Z_0(\cA)\otimes Z_0(\cB)$. In general we may have $\partial^{(\cA)}c_1\neq 0$, but we are able to reduce to the case where $\partial^{(\cA)}c_1=0$ by replacing $c$ with $c'=c+\partial_{1}^{\cC}f$ for an appropriate choice of $f\in A_1\otimes B_0$. Specifically, we choose $f$ to ensure that $\partial^{(\cA)}c_1'=0$, and we bound $|f|$ in terms of the filling constant $\mu_0(\cB)$ (along with some other parameters). As long as $|f|$ and $w^{\cA}$ are sufficiently small, then a lower bound on $|c_0'|$ implies a lower bound on $|c_0|\geq|c_0'|-|\partial^{(\cA)}f|\geq|c_0'|-w^{\cA}|f|$, as desired. The formal proof details are presented below.

\begin{proof}[Proof of Theorem~\ref{thm:proddis}]
  Consider any $c\in Z_i(\cC)\setminus\spn\{(L^{\cC})^i\}^\perp$. Our goal is to bound $|c|$. By definition $c$ has the form
  \begin{equation*}
    c = (c_1,c_0,c_{-1}) \in (A_1\otimes B_{i-1}) \oplus (A_0\otimes B_i) \oplus (A_{-1}\otimes B_{i+1}).
  \end{equation*}
  Letting $\partial^{(\cA)}=\partial^{\cA}\otimes I$ and $\partial^{(\cB)}=I\otimes\partial^{\cB}$ as above, then by definition
  \begin{equation}
    \label{eq:cycledef}
    \partial^{\cC}c = (-\partial^{(\cB)}c_1,\; \partial^{(\cA)}c_1+\partial^{(\cB)}c_0,\; \partial^{(\cA)}c_0-\partial^{(\cB)}c_{-1}) = (0,0,0).
  \end{equation}
  Therefore $c_1\in A_1\otimes Z_{i-1}(\cB)\subseteq A_1\otimes B_{i-1}$ is a matrix in which every row lies in $Z_{i-1}(\cB)$.

  Let $h=\dim H_{i-1}(\cB)$, and fix vectors $b_1,\dots,b_h\in Z_{i-1}(\cB)$ that provide representatives for a basis of $H_{i-1}(\cB)$. Then there exists a unique decomposition
  \begin{equation}
    \label{eq:c1decomp}
    c_1 = \gamma_0 + \gamma_1 + \cdots + \gamma_h,
  \end{equation}
  where\footnote{See Footnote~\ref{footnote:Bclash}.} $\gamma_0\in A_1\otimes B_{i-1}(\cB)$, and for $j\in[h]$ then $\gamma_j\in A_1\otimes\spn\{b_j\}$. Note that
  \begin{equation}
    \label{eq:gamma0bound}
    |\gamma_0| \leq \frac{n^{\cB}_{i-1}}{d_{i-1}(\cB)}\cdot|c_1|,
  \end{equation}
  as for any $j$ such that the $j$th row $(c_1)_j$ of $c_1$ does not equal the $j$th row $(\gamma_0)_j$ of $\gamma_0$, then $(\gamma_1+\cdots+\gamma_h)_j\neq 0$, which implies that $(c_1)_j\in Z_{i-1}(\cB)\setminus B_{i-1}(\cB)$ and therefore $|(c_1)_j|\geq d_{i-1}(\cB)$. Therefore $|(\gamma_0)_j|\leq(n^{\cB}_{i-1}/d_{i-1}(\cB))\cdot|(c_1)_j|$, so summing over all $j$ yields~(\ref{eq:gamma0bound}).

  By~(\ref{eq:cycledef}), we have $\partial^{(\cA)}c_1\in A_0\otimes B_{i-1}(\cB)$. But by definition $\partial^{(\cA)}\gamma_0\in A_0\otimes B_{i-1}(\cB)$ while for $j\in[h]$ we have $\partial^{(\cA)}\gamma_j\in A_0\otimes\spn\{b_j\}$, where $b_1,\dots,b_h$ form a basis for $Z_{i-1}(\cB)/B_{i-1}(\cB)$. Therefore we must in fact have
  \begin{equation}
    \label{eq:pgj0}
    \partial^{(\cA)}\gamma_j=0 \hspace{1em} \forall j\in[h],
  \end{equation}
  as otherwise some row of $\partial^{(\cA)}c_1$ would lie outside of $B_{i-1}(\cB)$.

  Applying the definition of the filling constant $\mu_i(\cB)$ to each row of the matrix $\gamma_0$, there exists some $f\in A_1\otimes B_i$ such that $\partial^{(\cB)}f=\gamma_0$ and $|f|\leq\mu_i(\cB)|\gamma_0|$. Define $c_0'\in A_0\otimes B_i$ by $c_0'=c_0+\partial^{(\cA)}f$. Then
  \begin{equation*}
    \partial^{(\cB)}c_0' = \partial^{(\cB)}c_0+\partial^{(\cA)}\partial^{(\cB)}f = -\partial^{(\cA)}c_1+\partial^{(\cA)}\gamma_0 = 0,
  \end{equation*}
  where the second equality above holds by~(\ref{eq:cycledef}) along with the definition of $f$, and the third equality holds by~(\ref{eq:c1decomp}) and~(\ref{eq:pgj0}). Thus $c_0'\in A_0\otimes Z_i(\cB)$ is a matrix in which every row lies in $Z_i(\cB)$.

  By the assumption that $c\notin\spn\{(L^{\cC})^i\}^\perp$, there exists some $a^*\otimes b^*\in (L^{\cC})^i$, where $a^*\in(L^{\cA})^0$ and $b^*\in(L^{\cB})^i$, such that
  \begin{equation*}
    (a^*\otimes b^*)\cdot c = (a^*\otimes b^*)\cdot c_0 = (a^*\otimes b^*)^\top c_0\neq 0.
  \end{equation*}
  Because $a^*\in Z^0(\cA)=B_0(\cA)^\perp$, it follows that $a^*$ is orthogonal to each column of $\partial^{(\cA)}f$, so $(a^*\otimes b^*)^\top c_0'=(a^*\otimes b^*)^\top c_0\neq 0$.

  Now by~(\ref{eq:cycledef}), we have $\partial^{(\cA)}c_0'=\partial^{(\cA)}c_0=\partial^{(\cB)}c_{-1}$. Thus as $b^*\in Z^i(\cB)=B_i(\cB)^\perp$, every row of $\partial^{(\cA)}c_0'$ is orthogonal to $b^*$, so $(\partial^{\cA}\otimes{b^*}^\top)c_0'=0$, meaning that $(I\otimes b^*)^\top c_0'\in Z_0(\cA)$. But because $(a^*\otimes b^*)^\top c_0'\neq 0$, we must in fact have $(I\otimes b^*)^\top c_0'\in Z_0(\cA)\setminus\spn\{(L^{\cA})^0\}^\perp$. Therefore
  \begin{equation*}
    |(I\otimes b^*)^\top c_0'| \geq d_0(\cA,L^{\cA}).
  \end{equation*}
  For each $j\in\supp((I\otimes b^*)^\top c_0')$, then the $j$th row $(c_0')_j$ of $c_0'$ is an element of $Z_i(\cB)$ that is not orthogonal to $b^*$, and thus $(c_0')_j\in Z_i(\cB)\setminus\spn\{(L^{\cB})^i\}^\perp$, which implies that
  \begin{equation*}
    |(c_0')_j| \geq d_i(\cB,L^{\cB}).
  \end{equation*}
  Combining the two inequalities above, we conclude that
  \begin{equation}
    \label{eq:c0plower}
    |c_0'| \geq d_0(\cA,L^{\cA})d_i(\cB,L^{\cB}).
  \end{equation}

  Now we have
  \begin{align}
    \label{eq:c0pupper}
    \begin{split}
      |c_0'|
      &\leq |c_0|+|\partial^{(\cA)}f| \\
      &\leq |c_0|+w^{\cA}|f| \\
      &\leq |c_0|+w^{\cA}\mu_i(\cB)|\gamma_0| \\
      &\leq |c_0|+w^{\cA}\mu_i(\cB)\frac{n^{\cB}_{i-1}}{d_{i-1}(\cB)}\cdot|c_1|,
    \end{split}
  \end{align}
  where the first three inequalities above hold by the definitions of the variables $c_0',f,w^{\cA},\gamma_0$, and the final inequality follows from~(\ref{eq:gamma0bound}). Combining~(\ref{eq:c0pupper}) with~(\ref{eq:c0plower}), and recalling that $|c|=|c_1|+|c_0|+|c_{-1}|$, we conclude the desired inequality
  \begin{equation*}
    |c| \geq \frac{d_0(\cA,L^{\cA})d_i(\cB,L^{\cB})}{w^{\cA}\mu_i(\cB) \cdot n_{i-1}^{\cB}/d_{i-1}(\cB)} \ .\qedhere
  \end{equation*}
\end{proof}

Theorem~\ref{thm:proddis} above analyzes how systolic distance behaves under homological products; Theorem~\ref{thm:prodfill} below uses similar techniques to analyze the behavior of filling constants. Although we will not apply Theorem~\ref{thm:prodfill} to any concrete codes in this paper, we present it for completeness, as it may be useful for future constructions of qLTCs based on homological products.

Note that filling constants do not depend on subsystem data $L$, so we do not consider subsystem complexes below. As a result, the proof below is more similar to the proof of Theorem~3.1 in \cite{kaufman_new_2021}. However, we still present the proof for completeness, as there are some small differences.

\begin{theorem}
  \label{thm:prodfill}
  Let
  \begin{equation*}
    \cA = (A_1 \xrightarrow{\partial^{\cA}_1} A_0 \xrightarrow{\partial^{\cA}_0} A_{-1})
  \end{equation*}
  be a 3-term chain complex, and let 
  \begin{equation*}
    \cB = (\cdots B_{i+2} \xrightarrow{\partial^{\cB}_{i+2}} B_{i+1} \xrightarrow{\partial^{\cB}_{i+1}} B_i \xrightarrow{\partial^{\cB}_{i}} B_{i-1} \xrightarrow{\partial^{\cB}_{i-1}} B_{i-2} \cdots)
  \end{equation*}
  be a $\geq 5$-term chain complex. Then the homological product $\cC=\cA\otimes\cB$ has $i$-filling constant
  \begin{equation}
    \label{eq:prodfill}
    \mu_i(\cC) \leq 2 \cdot \mu_{i-1}(\cB)w^{\cA}\frac{n^{\cB}_{i-1}}{d_{i-1}(\cB)} \cdot \left(M_1(\cA)+\mu_i(\cB)w^{\cA}\frac{n^{\cB}_i}{d_i(\cB)}\cdot(M_0(\cA)+\mu_{i+1}(\cB))\right)
  \end{equation}
  and $i$-collective filling constant
  \begin{equation}
    \label{eq:prodcolfill}
    M_i(\cC) \leq 2 \cdot M_{i-1}(\cB)w^{\cA}\frac{n^{\cB}_{i-1}}{d_{i-1}(\cB)} \cdot \left(M_1(\cA)+M_i(\cB)w^{\cA}\frac{n^{\cB}_i}{d_i(\cB)}\cdot(M_0(\cA)+M_{i+1}(\cB))\right),
  \end{equation}
  where $n^{\cB}_j=\dim B_j$.
\end{theorem}
\begin{proof}
  We first prove~(\ref{eq:prodfill}). Consider any $c\in B_{i-1}(\cC)$. Our goal is to find a low-weight filling $F\in C_i$ for $c$. Write
  \begin{equation*}
    c = (c_1,c_0,c_{-1}) \in (A_1\otimes B_{i-2}) \oplus (A_0\otimes B_{i-1}) \oplus (A_{-1}\otimes B_i)
  \end{equation*}
  and define $\partial^{(\cA)},\partial^{(\cB)}$ as in the proof of Theorem~\ref{thm:proddis}; note that again, Equation~(\ref{eq:cycledef}) holds for $c$. By definition, there exists some $e=(e_1,e_0,e_{-1})\in C_i$ such that
  \begin{equation*}
    \partial^{\cC}e = (-\partial^{(\cB)}e_1,\; \partial^{(\cA)}e_1+\partial^{(\cB)}e_0,\; \partial^{(\cA)}e_0-\partial^{(\cB)}e_{-1}) = (c_1,c_0,c_{-1}).
  \end{equation*}
  Therefore $c_1\in A_1\otimes B_{i-2}(\cB)$, so applying the definition of $\mu_{i-1}(\cB)$ to each row of $c_1$, there exists some $f_1\in A_1\otimes B_{i-1}$ such that $-\partial^{(\cB)}f_1=c_1$ and $|f_1|\leq\mu_{i-1}(\cB)|c_1|$. Let $c_0'=c_0-\partial^{(\cA)}f_1$. Then
  \begin{equation*}
    \partial^{(\cB)}c_0' = \partial^{(\cB)}c_0-\partial^{(\cA)}\partial^{(\cB)}f = \partial^{(\cB)}c_0+\partial^{(\cA)}c_1 = 0,
  \end{equation*}
  where the final equality above holds by~(\ref{eq:cycledef}), so $c_0'\in A_0\otimes Z_{i-1}(\cB)$

  Now for $j\in\{-1,0\}$, let $h_j=\dim H_{i-j-1}(\cB)$, and let $b^{j,1},\dots,b^{j,h_j}\in Z^{i-j-1}(\cB)$ and $b_{j,1},\dots,b_{j,h_j}\in Z_{i-j-1}(\cB)$ be dual bases for $H^{i-j-1}(\cB)$ and $H_{i-j-1}(\cB)$ respectively.

  By definition there exists a unique decomopsition
  \begin{equation*}
    c_0' = \gamma_{0,0}+\gamma_{0,1}+\cdots+\gamma_{0,h_0},
  \end{equation*}
  where $\gamma_{0,0}\in A_0\otimes B_{i-1}(\cB)$, and for $j\in[h_0]$ then $\gamma_{0,j}\in A_0\otimes\spn\{b_{0,j}\}$. In fact, by the definition of dual cohomology/homology bases, for all $j\in[h_0]$ we have $\gamma_{0,j}=\eta_{0,j}\otimes b_{0,j}$ for $\eta_{0,j}:=(I\otimes b^{0,j})^\top c_0'$. Because $c_0'=\partial^{(\cA)}(e_1-f_1)+\partial^{(\cB)}e_0$, applying $(I\otimes b^{0,j})^\top$ yields that $\eta_{0,j}=(\partial^{\cA}\otimes {b^{0,j}}^\top)(e_1-f_1)\in B_0(\cA)$ for all $j\in[h_0]$. Therefore by the definition of the collective filling constant $M_1(\cA)$, there exist some $\phi_{0,j}\in A_1$ for $j\in[h_0]$ satisfying $\partial^{(\cA)}\phi_{0,j}=\eta_{0,j}$ and
  \begin{equation*}
    \left|\bigcup_{j\in[h_0]}\phi_{0,j}\right| \leq M_1(\cA)\left|\bigcup_{j\in[h_0]}\eta_{0,j}\right|.
  \end{equation*}
  Define $f_1'\in A_1\otimes Z_{i-1}(\cB)$ by $f_1'=\sum_{j\in[h_0]}\phi_{0,j}\otimes b_{0,j}$. Then $\partial^{(\cA)}f_1'=\sum_{j\in[h_0]}\gamma_{0,j}$, and
  \begin{equation*}
    |f_1'| \leq \left|\bigcup_{j\in[h_0]}\phi_{0,j}\right|\cdot n^{\cB}_{i-1} \leq M_1(\cA)\left|\bigcup_{j\in[h_0]}\eta_{0,j}\right|\cdot n^{\cB}_{i-1} \leq M_1(\cA)\cdot\frac{|c_0'|}{d_{i-1}(\cB)}\cdot n^{\cB}_{i-1},
  \end{equation*}
  where the final inequality above holds because for every $\ell$ such that $(\eta_{0,j})_\ell\neq 0$ for some $j\in[h_0]$, then the $\ell$th row of $c_0'$ by definition lies in $Z_{i-1}(\cB)\setminus B_{i-1}(\cB)$, and thus has weight $\geq d_{i-1}(\cB)$.

  Define $F_1:=f_1+f_1'\in A_1\otimes B_{i-1}$. Observe that $-\partial^{(\cB)}F_1=-\partial^{(\cB)}f_1=c_1$ and $c_0-\partial^{(\cA)}F_1=c_0'-\partial^{(\cA)}f_1'=\gamma_{0,0}$. Therefore $c=(0,\gamma_{0,0},c_{-1})+\partial^{\cC}(F_1,0,0)$, so in particular $(0,\gamma_{0,0},c_{-1})\in B_{i-1}(\cC)$. Furthermore,
  \begin{align}
    \label{eq:F1bound}
    \begin{split}
      |F_1|
      &\leq |f_1|+|f_1'| \\
      &\leq |f_1|+M_1(\cA)\cdot\frac{n^{\cB}_{i-1}}{d_{i-1}(\cB)}(|c_0|+w^{\cA}|f_1|) \\
      &\leq \mu_{i-1}(\cB)|c_1|+M_1(\cA)\cdot\frac{n^{\cB}_{i-1}}{d_{i-1}(\cB)}(|c_0|+w^{\cA}\mu_{i-1}(\cB)|c_1|),
    \end{split}
  \end{align}
  where the second inequality above applies the definition of $c_0'=c_0-\partial^{(\cA)}f_1$.

  We now construct an element $F_0\in A_0\otimes B_i$ analogously to the construction of $F_1$ above, but ``shifted down'' one level in $\cA$ and ``shifted up'' one level in $\cB$; specifically, $\gamma_{0,0}$ will replace the role that $c_1$ played above, while $c_{-1}$ will replace the role that $c_0$ played above.

  In more detail, to construct $F_0$, because each row of $\gamma_{0,0}$ by construction lies in $B_{i-1}(\cB)$, there exists some $f_0\in A_0\otimes B_i$ such that $\partial^{(\cB)}f_0=\gamma_{0,0}$ and
  \begin{equation}
    \label{eq:f0bound}
    |f_0| \leq \mu_i(\cB)|\gamma_{0,0}| \leq \mu_i(\cB)\cdot\frac{n^{\cB}_{i-1}}{d_{i-1}(\cB)}\cdot|c_0'| \leq \mu_i(\cB)\cdot\frac{n^{\cB}_{i-1}}{d_{i-1}(\cB)}(|c_0|+w^{\cA}\mu_{i-1}(\cB)|c_1|),
  \end{equation}
  where the second inequality above holds because for every $\ell$ such that the $\ell$th rows $(\gamma_{0,0})_\ell\neq(c_0')_\ell$ are distinct, then $(c_0')_\ell$ lies in $Z_{i-1}(\cB)\setminus B_{i-1}(\cB)$ and thus has weight $\geq d_{i-1}(\cB)$.
  
  Let $c_{-1}'=c_{-1}-\partial^{(\cA)}f_0$. Then $\partial^{(\cB)}c_{-1}'=0$ because $\partial^{(\cA)}\gamma_{0,0}-\partial^{(\cB)}c_{-1}=0$, which in turn holds because $\partial^{\cC}(0,\gamma_{0,0},c_{-1})=(0,0,0)$. Thus there exists a unique decomposition
  \begin{equation*}
    c_{-1}' = \gamma_{-1,0}+\gamma_{-1,1}+\cdots+\gamma_{-1,h_i}
  \end{equation*}
  where $\gamma_{-1,0}\in A_{-1}\otimes B_i(\cB)$, and for $j\in[h_{-1}]$ then $\gamma_{-1,j}\in A_{-1}\otimes\spn\{b_{-1,j}\}$.

  Similarly as for $c_0'$ above, for $j\in[h_{-1}]$ we have $\gamma_{-1,j}=\eta_{-1,j}\otimes b_{-1,j}$ for $\eta_{-1,j}:=(I\otimes b^{-1,j})^\top c_{-1}'$, and $\eta_{-1,j}\in B_{-1}(\cA)$. Therefore there exist some $\phi_{-1,j}\in A_0$ for $j\in[h_{-1}]$ satisfying $\partial^{(\cA)}\phi_{-1,j}=\eta_{-1,j}$ and
  \begin{equation*}
    \left|\bigcup_{j\in[h_{-1}]}\phi_{-1,j}\right| \leq M_0(\cA)\left|\bigcup_{j\in[h_{-1}]}\eta_{-1,j}\right|.
  \end{equation*}
  Define $f_0'\in A_0\otimes Z_i(\cB)$ by $f_0'=\sum_{j\in[h_{-1}]}\phi_{-1,j}\otimes b_{-1,j}$, so that $\partial^{(\cA)}f_0'=\sum_{j\in[h_{-1}]}\gamma_{-1,j}$, and
  \begin{equation*}
    |f_0'| \leq M_0(\cA) \cdot \frac{|c_{-1}'|}{d_i(\cB)} \cdot n^{\cB}_i.
  \end{equation*}

  Define $F_0:=f_0+f_0'\in A_0\otimes B_i$. By construction $\partial^{\cA}F_0=c_{-1}-\gamma_{-1,0}$ and $\partial^{\cB}F_0=\gamma_{0,0}$. Thus $c=(0,0,\gamma_{-1,0})+\partial^{\cC}(F_1,F_0,0)$. Furthermore,
  \begin{align}
    \label{eq:F0bound}
    \begin{split}
      |F_0|
      &\leq |f_0|+|f_0'| \\
      &\leq |f_0|+M_0(\cA)\cdot\frac{n^{\cB}_i}{d_i(\cB)}(|c_{-1}|+w^{\cA}|f_0|),
    \end{split}
  \end{align}
  where $|f_0|$ is bounded by~(\ref{eq:f0bound})

  Finally, because $\gamma_{-1,0}\in A_{-1}\otimes B_i(\cB)$, there exists some $F_{-1}\in A_{i+1}\otimes B_{-1}$ such that $\partial^{(\cB)}F_{-1}=\gamma_{-1,0}$ and
  \begin{align}
    \label{eq:Fm1bound}
    \begin{split}
      |F_{-1}|
      &\leq \mu_{i+1}(\cB)|\gamma_{-1,0}| \\
      &\leq \mu_{i+1}(\cB)\cdot\frac{n^{\cB}_i}{d_i(\cB)}\cdot|c_{-1}'| \\
      &\leq \mu_{i+1}(\cB)\cdot\frac{n^{\cB}_i}{d_i(\cB)}(|c_{-1}|+w^{\cA}|f_0|),
    \end{split}
  \end{align}
  where $|f_0|$ is given by~(\ref{eq:f0bound}).

  Letting $F:=(F_1,F_0,F_{-1})$, we see that $c=\partial^{\cC}F$. Combining~(\ref{eq:F1bound}), (\ref{eq:F0bound}), (\ref{eq:Fm1bound}), and~(\ref{eq:f0bound}) gives
  \begin{align*}
    |F|
    &= |F_1|+|F_0|+|F_{-1}| \\
    &\leq \left(M_1(\cA)\cdot\frac{n^{\cB}_{i-1}}{d_{i-1}(\cB)}\cdot w^{\cA}+1\right)\mu_{i-1}(\cB)|c| \\
    &\hspace{1em} + \left((M_0(\cA)+\mu_{i+1}(\cB))\frac{n^{\cB}_i}{d_i(\cB)}\cdot w^{\cA}+1\right)(|c_{-1}|+|f_0|) \\
    &\leq \left(M_1(\cA)\cdot\frac{n^{\cB}_{i-1}}{d_{i-1}(\cB)}\cdot w^{\cA}+1\right)\mu_{i-1}(\cB)|c| \\
    &\hspace{1em} + \left((M_0(\cA)+\mu_{i+1}(\cB))\frac{n^{\cB}_i}{d_i(\cB)}\cdot w^{\cA}+1\right) \mu_i(\cB)\cdot\frac{n^{\cB}_{i-1}}{d_{i-1}(\cB)}\cdot w^{\cA}\mu_{i-1}(\cB)|c| \\
    &\leq 2 \cdot \mu_{i-1}(\cB)w^{\cA}\frac{n^{\cB}_{i-1}}{d_{i-1}(\cB)} \cdot \left(M_1(\cA)+\mu_i(\cB)w^{\cA}\frac{n^{\cB}_i}{d_i(\cB)}\cdot(M_0(\cA)+\mu_{i+1}(\cB))\right)\cdot|c|,
  \end{align*}
  which completes the proof of~(\ref{eq:prodfill}).

  To prove~(\ref{eq:prodcolfill}), we simply repeat the proof of~(\ref{eq:prodfill}) above, but we replace vectors with collections of vectors, and filling constants with collective filling constants. Specifically, instead of being given a boundary $c=(c_1,c_0,c_{-1})\in B_{i-1}(\cC)$, we assume that for some $m\in\bN$ we are given a collection of boundaries $c^\ell=(c^\ell_1,c^\ell_0,c^\ell_{-1})\in B_{i-1}(\cC)$ for $\ell\in[m]$. Our goal is to find fillings $F^\ell\in C_i$ so that each $c^\ell=\partial^{\cC}F^\ell$ and $|\bigcup_{\ell\in[m]}F^\ell|$ is small. To find such $F^\ell$, we simply repeat the above proof of~(\ref{eq:prodfill}), but we replace each of the vectors $e=(e_1,e_0,e_{-1}),f_1,c_0',\gamma_{0,j},\eta_{0,j},\phi_{0,j},f_1',f_0,c_{-1}',\gamma_{-1,j},\eta_{-1,j},\phi_{-1,j},f_0',F=(F_1,F_0,F_{-1})$ with a collection of $m$ vectors, and we replace every filling constant $\mu_k$ with the respective collective filling constant $M_k$. All Hamming weights of vectors are replaced with Hamming weights of unions over the $m$ vectors in the respective collection. For more details, the reader can also refer to the proof of Theorem~3.1 of \cite{kaufman_new_2021}, which is similar.
\end{proof}

\subsection{Sample Application to Product-Expansion-Based Codes of Section~\ref{sec:peprod}}
In the example below, we provide a sample application of Theorem~\ref{thm:proddis} to the codes from Section~\ref{sec:peprod}, in order to construct $[[N,\Theta(N),\Omega(N^{3/4})]]$ codes with locality $O(N^{1/4})$. Specifically, we obtain such codes by taking multi-sector homological products of the codes in Corollary~\ref{cor:ssperandom} or Corollary~\ref{cor:sspeRS}, which in turn were obtained as single-sector homological products of random codes or quantum Reed-Solomon codes respectively.

\begin{example}
  Let $\cA=\cC_1\otimes\cC_2$ be a single-sector chain complex given by either Corollary~\ref{cor:ssperandom} or Corollary~\ref{cor:sspeRS} for some fixed parameters $\epsilon,R_1,R_2>0$, so that the quantum code associated to $\cA$ has parameters $[[n^2,\Theta(n^2),\Theta(n^2)]]_q$ with locality $O(n)$ and soundness $\Omega(1)$. Here the constants hidden by $\Theta,O,\Omega$ depend only on $\epsilon,R_1,R_2$.

  Note that if obtain $\cA$ from Corollary~\ref{cor:ssperandom}, so that $\cC_1,\cC_2$ are obtained from random codes, then $q$ can be any power of $2$, regardless of the value of $N$. If we instead take $\cA$ from Corollary~\ref{cor:sspeRS}, so that $\cC_1,\cC_2$ are obtained from Reed-Solomon codes, then $q$ can again be any power of $2$, but we must have $n=q$.

  Now for $t\in\bN$ let
  \begin{equation*}
    \cA_t = \left(A \xrightarrow{\partial^{\cA}} A \xrightarrow{\partial^{\cA}} \cdots \xrightarrow{\partial^{\cA}} A\right)
  \end{equation*}
  be the $t$-term (multi-sector) chain complex obtained from the single-sector chain complex $\cA$ in the natural way, by simply stringing together $t-1$ copies of the boundary map of $\cA$. Also let $L^{\cA}=((L^{\cA})^*,(L^{\cA})_*)$ be a pair of dual cohomology/homology bases for $\cA$.

  Then by Theorem~\ref{thm:proddis}, the quantum code associated to the (middle 3 terms) of the (multi-sector subsystem) homological product $(\cA_3,L^{\cA})\otimes(\cA_5,L^{\cA})$ has parameters $[[\Theta(n^4),\Theta(n^4),\Omega(n^3)]]_q$ with locality $O(n)$. Note that here we associate the subsystem data $L^{\cA}$ with the middle term of $\cA_3$ and of $\cA_5$.
\end{example}

\subsection{Iterative Application to QLTCs using Locality Reduction}
\label{sec:iterapp}
In this section, we apply the results of Section~\ref{sec:genbounds} along with the weight (i.e.~locality) reduction results of \cite{hastings_quantum_2023} to chain complexes from \cite{dinur_expansion_2024,kalachev_personal_2024}. Specifically, starting with a single appropriate constant-sized chain complex, such as an instance of the \cite{dinur_expansion_2024,kalachev_personal_2024} construction, we are able to iteratively construct an infinite family of qLDPC (subsystem) codes with close to linear distance from smaller ones.

Below, we state the main result of \cite{dinur_expansion_2024}, which is proven using the result of \cite{kalachev_personal_2024} stated in Theorem~\ref{thm:petrand}.

\begin{theorem}[Follows from \cite{dinur_expansion_2024,kalachev_personal_2024}]
  \label{thm:dlv}
  For every fixed integer $t\geq 2$, there is an explicit family of $(t+1)$-term chain complexes $(\cC_N)_{N\in\bN}$ over $\bF_2$ such that each
  \begin{equation*}
    \cC_N = (C_{N,t} \xrightarrow{\partial_t} C_{N,t-1} \xrightarrow{\partial_{t-1}} \cdots \xrightarrow{\partial_1} C_{N,0})
  \end{equation*}
  satisfies
  \begin{align}
    w^{\cC_N} &= O(1) \\
    \dim C_{N,i} &= \Theta(N) \hspace{1em} \forall\; 0\leq i\leq t \\
    d_i(\cC_N) &\geq \Omega(N/(\log N)^{t-1}) \hspace{1em} \forall\; 1\leq i\leq t \\
    d^i(\cC_N) &\geq \Omega(N/(\log N)^{t-1}) \hspace{1em} \forall\; 0\leq i\leq t-1 \\
    \label{eq:dlvmudi} \mu_i(\cC_N) &\leq O(\log N)^{t-1} \hspace{1em} \forall\; 2\leq i\leq t \\
    \label{eq:dlvmuui} \mu^i(\cC_N) &\leq O(\log N)^{t-1} \hspace{1em} \forall\; 0\leq i\leq t-2,
  \end{align}
  where the constants hidden by the $\Theta,\Omega,O$ above may depend on $t$.
\end{theorem}

We will apply Theorem~\ref{thm:proddis} to the chain complexes in Theorem~\ref{thm:dlv}. We will, however, not attempt to apply Theorem~\ref{thm:prodfill} to these chain complexes in this paper; in order to do so, we would want analogous bounds as in~(\ref{eq:dlvmudi}) and~(\ref{eq:dlvmuui}) to also hold for \emph{collective} (co)filling constants $M_i$ and $M^i$. Later in Remark~\ref{remark:collfill}, we explain why such bounds indeed hold.

It is tempting to attempt iteratively constructing qLDPC codes as follows. Fix some large constant-sized 5-term chain complex $\cB$ from Theorem~\ref{thm:dlv}, and re-index the nonzero terms to range from $-2,\dots,2$. Let $L_0=((L_0)^0,(L_0)_0)$ consist of dual 0-cohomology/homology bases. Then let $\cA_0$ be the truncation of $\cB$ to the middle 3 terms, and for $i\in\bN$, define a 0-subsystem complex $(\cA_i,L_i)$ as the truncation of $(\cA_{i-1},L_{i-1})\otimes(\cB,L_0)$ to the middle 3 terms. We may then try to bound $d_0(\cA_i,L_i)$ and $d^0(\cA_i,L_i)$ using Theorem~\ref{thm:proddis}. However, the locality of $\cA_i$ is roughly $w^{\cA_i}\approx(i+1)w^{\cB}$, which for sufficiently large $i$ is greater $d_0(\cB,L_0),d^0(\cB,L_0)=O(1)$, as $\cB$ is constant-sized. Therefore for all sufficiently large $i$, the distance lower bound for $\cA_i$ from Theorem~\ref{thm:proddis} decreases as $i$ increases. In contrast, we want a distance lower bound that increases, ideally linearly, with $\dim(A_i^0)$; note that $\dim(A_i^0)$ grows exponentially with $i$.

The above attempt at iteratively constructing large-distance qLDPC codes failed because the locality grows in each iteration. We resolve this issue by applying the locality-reduction techniques of \cite{hastings_quantum_2023} (see also \cite{wills_tradeoff_2024}). Below, we state a locality-reduction result that follows from \cite{hastings_quantum_2023} with parameters sufficient for our purposes; we emphasize that we do not need the tightest possible bounds.

\begin{theorem}[Follows from \cite{hastings_quantum_2023}]
  \label{thm:locred}
  There exists a constant $c_0\in\bN$ and a monotonically increasing polynomial $p_0(\cdot)$ such that the following holds. Let $(Q,L)$ be a subsystem code of length $n$, dimension $k$, locality $w$, distance $d\geq p_0(w)$ (see Definition~\ref{def:subsystem}), and alphabet size $q=2$. Then there exists another subsystem code $(\tilde{Q},\tilde{L})$ with length $\tilde{n}\leq np_0(w)$, dimension $\tilde{k}=k$, locality $\tilde{w}\leq c_0$, and distance $\tilde{d}\geq d/p_0(w)$.\footnote{In Theorem~1 of \cite{hastings_quantum_2023}, these bounds on $\tilde{n}$ and $\tilde{d}$ are stated with an additional $\poly(\log n)$ loss factor. However, as remarked in Section~3.4 of \cite{wills_tradeoff_2024}, such additional $\poly(\log n)$ factors are unnecessary, even when using the proof of \cite{hastings_quantum_2023}.}

  Furthermore, given $(Q,L)$, such $(\tilde{Q},\tilde{L})$ can be computed by a randomized algorithm with success probability $\geq 1-\exp(-\Omega(n))$ in time $\poly(n)$, or by a deterministic algorithm in time $\poly(n,\exp(\poly(w)))$.
\end{theorem}

\begin{remark}
  Theorem~\ref{thm:locred} is not exactly as stated in \cite{hastings_quantum_2023,wills_tradeoff_2024}, though our statement follows from essentially the same proof. This remark briefly describes the differences.

  First, \cite{hastings_quantum_2023} does not directly analyze the polynomial-time algorithms described in Theorem~\ref{thm:locred}, but the construction is naturally algorithmic, so tracing through the steps yields the stated algorithms. The only randomness in the $\poly(n)$-time algorithm is used in the application of the decongestion lemma of \cite{freedman_building_2021}; instead using a brute-force search yields the $\poly(n,\exp(\poly(w)))$-time deterministic algorithm.

  The requirement that $d\geq p_0(w)$ corresponds to the assumption that the relevant code is ``reasonable'' in the language of \cite{hastings_quantum_2023}. Although \cite{hastings_quantum_2023} shows that this assumption can be removed, we maintain it to emphasize that we do not need the strongest possible result, and to maintain consistency with the exposition of \cite{wills_tradeoff_2024}, which only considers reasonable codes.

  Finally, the main difference between Theorem~\ref{thm:locred} and the results in \cite{hastings_quantum_2023,wills_tradeoff_2024} is that Theorem~\ref{thm:locred} considers subsystem codes (with their associated notions of distance and dimension), whereas \cite{hastings_quantum_2023,wills_tradeoff_2024} only considered non-subsystem codes. However, tracing through the proof of \cite{hastings_quantum_2023} reveals that the construction goes through flawlessly for subsystem codes as well; the only modification is to restrict attention to bounding the weight of appropriate codewords (i.e.~logical operators), as specified by the subsystem data $L$.
\end{remark}

We now show how to construct an infinite family of qLDPC (subsystem) codes of close to linear distance and dimension by iteratively applying Theorem~\ref{thm:proddis} and Theorem~\ref{thm:locred}, using a constant-sized instance of the chain complexes in Theorem~\ref{thm:dlv} as a starting point. The iterative construction is formally presented in the following definition.

\begin{definition}
  \label{def:iterative}
  Given a parameter $N\in\bN$, we define a family of 3-term 0-subsystem chain complexes $(\cA_i,L_i)_{i\in\bZ_{\geq 0}}$ as follows. Let
  \begin{equation*}
    \cC_N = (C_{N,2} \rightarrow C_{N,1} \rightarrow C_{N,0} \rightarrow C_{N,-1} \rightarrow C_{N,-2})
  \end{equation*}
  be the $5$-term chain complex from Theorem~\ref{thm:dlv} with parameters $N$ and $t=4$, where we have reindexed the chain complex terms to range from $-2$ to $+2$. Let $L^{\cC_N}=((L^{\cC_N})^0,(L^{\cC_N})_0)$ be a pair of dual $0$-cohomology/homology bases for $\cC_N$.

  Let $(\cA_0,L_0)$ be the 3-term subsystem complex obtained by applying the locality-reduction transformation\footnote{\label{footnote:noredundantchecks} Here we view the locality-reduced code as a 3-term chain complex in the standard way, by letting $\delta_0=H_X$ and $\partial_0=H_Z$. We may furthermore assume that these maps have full rank, which can always be achieved by removing rows from $H_X$ (resp.~$H_Z$) that lie in the span of some other rows. Removing such redundant parity checks can only decrease (i.e.~improve) the locality, and does not affect the code's distance as $\ker H_X$ and $\ker H_Z$ remain unchanged.} in Theorem~\ref{thm:locred} to the subsystem code associated to (the middle 3 terms of) $(\cC_N,L^{\cC_N})$, where here we use the correspondence between 3-term subsystem chain complexes and subsystem CSS codes described in Definition~\ref{def:subsystemcomplex}.

  For $i\geq 1$, given $(\cA_{i-1},L_{i-1})$, we inductively define $(\cA_i,L_i)$ to be the 3-term subsystem chain complex obtained by applying the locality-reduction transformation in Theorem~\ref{thm:locred} to the subsystem code associated to (the middle 3 terms of) the homological product $(\cA_{i-1},L_{i-1})\otimes(\cC_N,L^{\cC_N})$.
\end{definition}

The following result bounds the code parameters of the construction in Definition~\ref{def:iterative}.

\begin{theorem}
  \label{thm:iterative}
  For every $\epsilon>0$, there exists a sufficiently large $N=N(\epsilon)\in\bN$ such that the $3$-term subsystem chain complexes $(\cA_i,L_i)$ with parameter $N$ defined in Definition~\ref{def:iterative} satisfy the following: letting $N_i:=\dim A_{i,0}$, then the subsystem code associated to $(\cA_i,L_i)$ has length $N_i$, dimension $\geq N_i^{1-\epsilon}$, distance $\geq N_i^{1-\epsilon}$, and locality $\leq O(1)$.

  Furthermore, for every fixed $\epsilon$ with $N=N(\epsilon)$, then $(\cA_i,L_i)$ can be constructed by a deterministic algorithm in time $\poly(N_i)$.
\end{theorem}
\begin{proof}
  Fix $\epsilon>0$, and define the constant $c_0$ and the polynomial $p_0(\cdot)$ as in Theorem~\ref{thm:locred}. For every $i$, let $k_i$ denote the dimension of the subsystem code associated to $(\cA_i,L_i)$.

  By Theorem~\ref{thm:dlv}, for all sufficiently large $N$ we have $d_0(\cC_N),d^0(\cC_N)\geq\Omega(N/\log^3N)>O(1)=p_0(c_0+w^{\cC_N})$. Therefore by Theorem~\ref{thm:dlv} and Theorem~\ref{thm:locred}, the subsystem code associated to $(\cA_0,L_0)$ has length $N_0\leq Np_0(w^{\cC_N})=O(N)$, locality $w^{\cA_0}\leq c_0=O(1)$, dimension $k_0=\Theta(N)\geq\Omega(N_0)$, and distance $d(\cA_0,L_0)\geq\Omega(N/\log^3N)/p_0(w^{\cC_N})\geq\Omega(N_0/\log^3N_0)$. Note that here $O,\Theta,\Omega$ hide absolute constants independent of $N$ and $\epsilon$. Thus for all $N$ sufficiently large, we have $k_0\geq N_0^{1-\epsilon}$ and $d(\cA_0,L_0)\geq N_0^{1-\epsilon}$.

  We now show the desired bounds on $(\cA_i,L_i)$ by induction. Assume that the subsystem code associated to $(\cA_{i-1},L_{i-1})$ has dimension $k_{i-1}\geq N_{i-1}^{1-\epsilon}$, distance $d(\cA_{i-1},L_{i-1})\geq N_{i-1}^{1-\epsilon}$, and locality $w^{\cA_{i-1}}\leq c_0$. Our goal is to show that if $N$ is greater than some sufficiently large constant depending only on $\epsilon$, then $(\cA_i,L_i)$ has dimension $k_i\geq N_i^{1-\epsilon}$, distance $d(\cA_i,L_i)\geq N_i^{1-\epsilon}$, and locality $w^{\cA_i}\leq c_0$.

  To begin, the homological product $(\cA_{i-1},L_{i-1})\otimes(\cC_N,L^{\cC_N})$ by construction has locality $\leq w^{\cA_{i-1}}+w^{\cC_N}\leq c_0+w^{\cC_N}$, and by Theorem~\ref{thm:proddis} satisfies the distance bounds
  \begin{align*}
    d_0((\cA_{i-1},L_{i-1})\otimes(\cC_N,L^{\cC_N})) &\geq \frac{d_0(\cA_{i-1},L_{i-1})d_0(\cC_N,L^{\cC_N})}{c_0\mu_0(\cC_N)\cdot\dim(\cC_{N,-1})/d_{-1}(\cC_N)} \\
    d^0((\cA_{i-1},L_{i-1})\otimes(\cC_N,L^{\cC_N})) &\geq \frac{d^0(\cA_{i-1},L_{i-1})d^0(\cC_N,L^{\cC_N})}{c_0\mu^0(\cC_N)\cdot\dim(\cC_{N,1})/d^{1}(\cC_N)}.
  \end{align*}
  By the inductive hypothesis along with Theorem~\ref{thm:dlv}, it follows that
  \begin{align}
    \label{eq:applyproddis}
    d((\cA_{i-1},L_{i-1})\otimes(\cC_N,L^{\cC_N}))
    &\geq \Omega\left(\frac{N_{i-1}^{1-\epsilon}\cdot N/\log^3N}{\log^3N\cdot\log^3N}\right) = \Omega\left(\frac{N_{i-1}^{1-\epsilon}\cdot N}{\log^9N}\right).
  \end{align}
  Furthermore, by definition the subsystem code associated to the product $(\cA_{i-1},L_{i-1})\otimes(\cC_N,L^{\cC_N})$ has dimension
  \begin{align*}
    k_i &= k_{i-1}\cdot\dim(H_0(\cC_N)) \geq \Omega(N_{i-1}^{1-\epsilon}\cdot N),
  \end{align*}
  and has length
  \begin{equation*}
    \dim(A_{i-1,-1})\dim(C_{N,1})+\dim(A_{i-1,0})\dim(C_{N,0})+\dim(A_{i-1,1})\dim(C_{N,-1}) \leq O(N_{i-1}\cdot N).
  \end{equation*}
  The final inequality above holds by Theorem~\ref{thm:dlv}, and because $\dim(A_{i-1,-1}),\dim(A_{i-1,1})\leq\dim(A_{i-1,0})=N_{i-1}$ by footnote~\ref{footnote:noredundantchecks}. Note that the $\Omega,O$ above hide absolute constants independent of $N,\epsilon,i$.

  For every $N$ larger than some absolute constant, the right hand side of~(\ref{eq:applyproddis}) is larger than $p(c_0+w^{\cC_N})$, as $c_0+w^{\cC_N}$ is by definition itself an absolute constant. Therefore we may apply Theorem~\ref{thm:locred} to the (middle 3 terms of) the homological product $(\cA_{i-1},L_{i-1})\otimes(\cC_N,L^{\cC_N})$, and we conclude that the resulting 3-term subsystem chain complex $(\cA_i,L_i)$ has length
  \begin{equation*}
    N_i \leq O(N_{i-1}\cdot N\cdot p_0(c_0+w^{\cC_N})) = O(N_{i-1}\cdot N),
  \end{equation*}
  dimension
  \begin{align*}
    k_i &\geq \Omega(N_{i-1}^{1-\epsilon}\cdot N) \geq \Omega(N_i^{1-\epsilon}\cdot N^\epsilon),
  \end{align*}
  distance
  \begin{equation*}
    d(\cA_i,L_i) \geq \Omega\left(\frac{N_{i-1}^{1-\epsilon}\cdot N}{\log^9N}\cdot\frac{1}{p_0(c_0+w^{\cC_N})}\right) \geq \Omega\left(N_i^{1-\epsilon}\cdot\frac{N^\epsilon}{\log^9N}\right).
  \end{equation*}
  and locality $w^{\cA_i}\leq c_0$, where in the inequalities above we absorb the absolute constant $p_0(c_0+w^{\cC_N})$ into those hidden by the $O,\Omega$. Thus if $N$ is greater than some absolute constant depending only on $\epsilon$ (and in particular, independent of $i$), then we have $k_i\geq N_i^{1-\epsilon}$ and $d(\cA_i,L_i)\geq N_i^{1-\epsilon}$, completing the inductive step, as desired.

  It remains to be shown that $(\cA_i,L_i)$ can be constructed by a deterministic algorithm in time $\poly(N_i)$. We again proceed by induction. For the base case, $(\cA_0,L_0)$ is a constant-sized object, and thus constructible in constant time. For the inductive step, given $(\cA_{i-1},L_{i-1})$, we can by definition compute the homological product $(\cA_{i-1},L_{i-1})\otimes(\cC_N,L^{\cC_N})$ in time $\poly(N_{i-1},N)=\poly(N_i)$, and by Theorem~\ref{thm:locred} we can then perform locality-reduction to obtain $(\cA_i,L_i)$ in time $\poly(N_{i-1}N,\exp(\poly(c_0+w^{\cC_N})))=\poly(N_i)$. Thus there is a deterministic algorithm to compute $(\cA_i,L_i)$ given $(\cA_{i-1},L_{i-1})$ in time $\poly(N_i)$. By construction $N_i\geq 2N_{i-1}$, so recursively applying this algorithm $i$ times to construct $(\cA_i,L_i)$ from $(\cA_0,L_0)$ yields an overall running time of $\poly(N_i)$.
\end{proof}

A natural extension of Theorem~\ref{thm:iterative} would be to construct qLTCs rather than qLDPC codes, by choosing $\cC_N$ to be a 7-term complex instead of a 5-term complex, and then replacing the application of Theorem~\ref{thm:proddis} with an application of Theorem~\ref{thm:prodfill}. However, the bounds in Theorem~\ref{thm:prodfill} depend on {\it collective} (co)filling constants, which are in general more difficult to bound than the (co)filling constants used in Theorem~\ref{thm:proddis}. Fortunately, Remark~\ref{remark:collfill} below explains why the bounds on the (co)filling constants in Theorem~\ref{thm:dlv} also apply to the collective (co)filling constants. Yet in the locality-reduction step in the proof of Theorem~\ref{thm:iterative}, we would also need to show that the collective (co)filling constants are preserved at all $3$ levels of the chain complex, up to a small loss. Specifically, we would need to extend the work of \cite{wills_tradeoff_2024}, which showed that locality-reduction preserves a quantum code's soundness, or equivalently, preserves the (co)filling constants in the middle level of the associated 3-term chain complex. We leave this direction for future work.

\begin{remark}
  \label{remark:collfill}
  \cite{dinur_expansion_2024} only claim to prove that the bounds~(\ref{eq:dlvmudi}) and~(\ref{eq:dlvmuui}) hold for the (co)filling constants $\mu_i(\cC_N)$ and $\mu^i(\cC_n)$ respectively. However, their proof implicitly implies that the analogous bounds also hold for the \textit{collective} (co)filling constants, so that (\ref{eq:dlvmudi}) and~(\ref{eq:dlvmuui}) in Theorem~\ref{thm:dlv} can be replaced by
  \begin{align}
    \label{eq:dlvMdi} M_i(\cC_N) &\leq O(\log N)^{t-1} \hspace{1em} \forall\; 2\leq i\leq t \\
    \label{eq:dlvMui} M^i(\cC_N) &\leq O(\log N)^{t-1} \hspace{1em} \forall\; 0\leq i\leq t-2.
  \end{align}
  Recall here that $\mu_i\leq M_i$ and $\mu^i\leq M^i$ by definition.

  Specifically, \cite{dinur_expansion_2024} first provide an interpretation of collective (co)filling in terms of the (co)filling constants of a ``direct product sheaf.'' In this interpretation, given a chain complex $\cC$ and some $m\in\bN$, the direct product sheaf $\cC^m$ has vector spaces $C_i^m$ and boundary maps $\partial_i^m$ that simply apply $\partial_i$ component-wise on each of the $m$ components. Each $C_i^m$ is also given a norm $|\cdot|_m$ that counts the number of coordinates in which any of the $m$ components is nonzero. The collective $i$-(co)filling constant of $\cC$ is then precisely the maximum over all $m\in\bN$ of the $i$-(co)filling constant of $\cC_m$ under the norm $|\cdot|_m$.

  At a high level, \cite{dinur_expansion_2024} prove Theorem~\ref{thm:dlv} (for the ordinary, non-collective (co)filling constants) by constructing certain chain complexes $\cC_N$ that they analyze via a local-to-global framework. Specifically, \cite{dinur_expansion_2024} show that the local structure of $\cC_N$ can be described by a constant-sized chain complex $\cL$ (see Section~4 of \cite{dinur_expansion_2024}). \cite{dinur_expansion_2024} show that the global complex $\cC_N$ satisfies the distance and expansion bounds in Theorem~\ref{thm:dlv} as long as the local complex $\cL$ has good \textit{(co)boundary expansion}, meaning that $\cL$ has vanishing (co)homology and has good (co)filling constants (where ``good'' means bounded by a constant independent of the dimension). \cite{dinur_expansion_2024} in turn then proves that $\cL$ satisfies these properties by applying the forthcoming result of Kalachev \& Panteleev stated in Theorem~\ref{thm:petrand}.

  To prove that $\cC_N$ furthermore has good \textit{collective} (co)filling constants as stated in~(\ref{eq:dlvMdi}) and~(\ref{eq:dlvMui}), it suffices to show that for all $m\in\bN$, the direct product sheaf $\cC_N^m$ has good (co)filling constants under the norm $|\cdot|_m$. The local-to-global results of \cite{dinur_expansion_2024} (Section~6 and Section~7 in their paper) translate flawlessly to the direct product sheaf $\cC_N^m$ with the norm $|\cdot|_m$; indeed, their entire paper is written using a similar sheaf terminology. Thus it suffices to show that the local direct product sheaf $\cL^m$ has (co)boundary expansion (i.e.~vanishing (co)homology and good (co)filling constants). But \cite{dinur_expansion_2024} show precisely this result in their Proposition~4.14. That is, \cite{dinur_expansion_2024} show collective (co)boundary expansion (or in their terminology, \textit{collective robustness}) of the local complex, which along with their local-to-global framework, implies~(\ref{eq:dlvMdi}) and~(\ref{eq:dlvMui}).
\end{remark}

\section{Conclusion}
In this paper, we provide new constructions of qLDPC codes of almost linear distance. Constructions of such objects were out of reach of known techniques until the breakthrough line of work \cite{hastings_fiber_2021,panteleev_quantum_2022,breuckmann_balanced_2021,panteleev_asymptotically_2022,leverrier_quantum_2022-1,dinur_good_2023} that developed over the past few years.\footnote{We remark that \cite{hastings_quantum_2023} applies locality-reduction to construct qLDPC codes of length $N$ and distance $\tilde{\Omega}(N^{2/3})$. Therefore these codes surpass the ``$\tilde{O}(\sqrt{N})$ distance barrier'' (see Section~\ref{sec:sqrtbarrier}), but do not approach linear distance.} The constructions from these works are based on a sort of balanced product (also called a lifted product) that applies to only certain codes with a particular group symmetry. In contrast, our constructions are based on iterative applications of the more basic homological product, which is generically defined with no group symmetries. Indeed, the classical analogue of our construction, namely iterated tensor products of classical codes, immediately gives classical LDPC codes of close-to-linear distance, and has also been shown to yield properties such as local testability \cite{ben-sasson_robust_2004}.

It is perhaps surprising that we are able to obtain similar parameters as in the classical case with just homological/tensor products, especially given that for many years, qLDPC codes constructed with similar techniques were not able to achieve distance greater than $\tilde{O}(\sqrt{N})$ (see Section~\ref{sec:sqrtbarrier}). We obtain our distance bounds by combining conceptual insights (e.g.~the use of subsystem codes in Section~\ref{sec:subprod}) with certain strong properties of the base codes. Specifically, our bounds assume that the starting codes either exhibit an appropriate product-expansion property (see Section~\ref{sec:peprod}), or are constant-sized qLTCs (see Section~\ref{sec:subprod}). Currently, the only known appropriate such qLTCs, given by \cite{dinur_expansion_2024,kalachev_personal_2024}, are themselves based on balanced products following the line of work described above. However, if we could for instance prove product-expansion of higher-order products of Reed-Solomon codes, we would be able to extend our results in Section~\ref{sec:peprod} to obtain a completely explicit construction of qLDPC codes of almost linear distance, without relying on any complex constant-sized objects. We leave such approaches for future work.



\bibliographystyle{alpha}
\bibliography{library}

\appendix

\section{Proofs of Product-Expansion Properties}
\label{sec:peproofs}
This section provides proofs for properties of product-expansion. For this purpose, we carry over the notation $C^{(i)}$, $C^{(i,j)}$, $|\cdot|_i$ from Section~\ref{sec:pe}.

We begin with a proof of Lemma~\ref{lem:homvanexp}.

\begin{proof}[Proof of Lemma~\ref{lem:homvanexp}]
  The lemma essentially follows from the discussion in Appendix~B of~\cite{kalachev_two-sided_2023}, though we include a proof for completeness. At a high level, \cite{kalachev_two-sided_2023} provide an interpretation of $C_1\boxplus\cdots\boxplus C_t$ using a homological product of chain complexes associated to the codes $C_1,\dots,C_t$. The lemma follows by applying the K\"{u}nneth formula to this product, and then tracing through the definitions.
  
  For $i\in[t]$, let $k_i=\dim C_i$, and define a 2-term cochain complex
  \begin{equation*}
    \cC_i=(\bF_q^{k_i}\xrightarrow{G_i}\bF_q^{n_i}),
  \end{equation*}
  where the coboundary map is a generator matrix $G_i\in\bF_q^{n_i\times k_i}$ for $C_i$, meaning that $C_i=\im G_i$. Then let
  \begin{equation*}
    \cA^* = (A^0 \xrightarrow{\delta_0} A^1 \xrightarrow{\delta_1} \cdots \xrightarrow{\delta_{t-1}} A^t)
  \end{equation*}
  be the $(t+1)$-term cochain complex given by the homological product (Definition~\ref{def:homprod}) of $\cC_1,\dots,\cC_t$, that is
  \begin{equation*}
    \cA = \cC_1 \otimes \cdots \otimes \cC_t.
  \end{equation*}
  The K\"{u}nneth formula (Proposition~\ref{prop:kunneth}) implies that the cohomology $H^i(\cA)$ vanishes for all $0\leq i\leq t-1$, so in particular $H^{t-1}(\cA)=0$.

  For $0\leq i\leq t$, the space $A^i$ is the direct sum of ${t\choose i}$ spaces of $t$-dimensional tensors, that is,
  \begin{equation}
    \label{eq:Aistructure}
    A^i \cong \bigoplus_{S\subseteq[t]:|S|=i}\left(\bigotimes_{j\in S}\bF_q^{n_j}\right)\otimes\left(\bigotimes_{j\in[t]\setminus S}\bF_q^{k_j}\right),
  \end{equation}
  where we write $\cong$ instead of $=$ above because the $t$ factors in the tensor products above should be reordered to go in increasing order by $j$. Letting $G^{(i)}=I^{\otimes i-1}\otimes G_i\otimes I^{t-i}$ denote the map that applies $G_i$ to all the direction-$i$ vectors in a $t$-dimensional tensor, then the coboundary maps of $\cA$ by definition consist of sums of maps $G^{(i)}$ with appropriate signs. More details on the chain complex $\cA$ can be found in Appendix~B of~\cite{kalachev_two-sided_2023} and in Section~4 of~\cite{dinur_expansion_2024} (though note that these works use different signing conventions for the coboundary maps).

  Now for $i\in[t]$, let $a_i={G^{(i)}}^{-1}(-1)^{i-1}(c_i-c_i')$ be the $t$-dimensional tensor obtained by applying $G_i^{-1}$ to every direction-$i$ column of $(-1)^{i-1}(c_i-c_i')$; note that as every direction-$i$ column of $c_i-c_i'$ lies in $C_i=\im G_i$, these inverses are (uniqely) well-defined. Then by definition $a:=(a_1,\dots,a_t)$ is an element of $A^{t-1}$, and
  \begin{equation*}
    \delta_{t-1}a = \sum_{i\in[t]}(-1)^{i-1}(-1)^{i-1}(c_i-c_i') = c-c = 0.
  \end{equation*}
  Thus because $H^{t-1}(\cA)=0$ as shown above, there must exist some $b\in A^{t-2}$ with $a=\delta_{t-2}b$. It follows from~(\ref{eq:Aistructure}) that $b$ is a collection of tensors $b=(b_{i,j})_{1\leq i<j\leq t}$, where $b_{i,j}$ has length $k_i$ and $k_j$ in directions $i$ and $j$ respectively, and length $n_\ell$ in all other directions $\ell\in[t]\setminus\{i,j\}$. Furthermore, for $i\in[t]$, by definition
  \begin{align*}
    a_i &= (\delta_{t-2}b)_i = \sum_{j=1}^{i-1}(-1)^{j-1}G^{(j)}b_{j,i} + \sum_{j=i+1}^t(-1)^{j-2}G^{(j)}b_{i,j}.
  \end{align*}
  Applying $(-1)^{i-1}G^{(i)}$ to both sides of the above equation, we conclude that
  \begin{equation*}
    c_i-c_i' = (-1)^{i-1}G^{(i)}a_i = \sum_{j=1}^{i-1}(-1)^{(j-1)+(i-1)}G^{(j)}G^{(i)}b_{j,i} + \sum_{j=i+1}^t(-1)^{(i-1)+(j-2)}G^{(i)}G^{(j)}b_{i,j}.
  \end{equation*}
  Thus for all $1\leq i<j\leq t$, if we let
  \begin{equation*}
    c_{i,j} = (-1)^{(i-1)+(j-1)}G^{(i)}G^{(j)}b_{i,j},
  \end{equation*}
  then $c_{i,j}\in C^{(i,j)}$ and~(\ref{eq:homvanexp}) holds, as desired.
\end{proof}

We now turn to proving Lemma~\ref{lem:pesubmain}. We begin with the following lemma.

\begin{lemma}
  \label{lem:pefewercodes}
  If a collection of codes $(C_i\subseteq\bF_q^{n_i})_{i\in[t]}$ has product-expansion $\rho>0$, and $C_t\neq\bF_q^{n_t}$, then the subcollection $(C_i)_{i\in[t-1]}$ also has product expansion at least $\rho$.
\end{lemma}
\begin{proof}
  Fix any $c\in C_1\boxplus\cdots\boxplus C_{t-1}$. Our goal is to construct a decomposition $c=c_1+\cdots+c_{t-1}$ with each $c_i\in C^{(i;t-1)}$ such that $\sum_{i\in[t-1]}n_i|c_i|_i\leq|c|/\rho$. For this purpose, because $C_t\subsetneq\bF_q^{n_t}$, the dual $C_t^\perp$ must be nonzero, so there exists some $v\in C_t^\perp\setminus\{0\}$. Fix an arbitrary coordinate $s\in\supp(v)\subseteq[n_t]$, and let $c'\in C_1\boxplus\cdots\boxplus C_t$ be natural the embedding of the $(t-1)$-dimensional tensor $c\in C_1\otimes\cdots\otimes C_{t-1}$ into $C_1\otimes\cdots\otimes C_{t-1}\otimes\bF_q^{\{s\}}\subseteq C_1\otimes\cdots\otimes C_t$. Formally, for $r\in[n_t]$, let $c'|_r$ denote the restriction of $c'$ to components where the last coordinate equals $r$. Then $c'|_s=c$ and $c'|_r=0$ for all $r\neq s$.

  By the product-expansion of $C_1,\dots,C_t$, there exists a decomposition $c'=c_1'+\cdots+c_t'$ with each $c_i'\in C^{(i;t)}$ such that $\sum_{i\in[t]}n_i|c_i'|_i\leq|c'|/\rho=|c|/\rho$. For $i\in[t]$, define
  \begin{equation*}
    c_i := (I^{\otimes t-1}\otimes v^\top)v_s^{-1}c_i'.
  \end{equation*}
  Then $c_t=0$, while $c_i\in C^{(i;t-1)}$ for $i\in[t-1]$. Thus
  \begin{equation*}
    c = (I^{\otimes t-1}\otimes v^\top)v_s^{-1}c' = c_1+\cdots+c_{t-1}.
  \end{equation*}
  Furthermore, by definition $|c_i|_i\leq|c_i'|_i$, so $\sum_{i\in[t-1]}n_i|c_i|_i\leq\sum_{i\in[t-1]}n_i|c_i'|_i\leq|c|/\rho$, as desired.
  
\end{proof}

We now apply Lemma~\ref{lem:pefewercodes} to show the following lemma, which will immediately imply Lemma~\ref{lem:pesubmain}.

\begin{lemma}
  \label{lem:pesubcode}
  Let $\bF_q$ be a field of characteristic $2$, and let $(C_i\subsetneq\bF_q^n)_{i\in[t]}$ be a collection of codes\footnote{We believe that the restrictions that $\bF_q$ have characteristic $2$ and that all $C_i$ have the same length $n$ are not necessary, but for simplicity in this paper we maintain these assumptions; see Remark~\ref{remark:dlvrestrict} for details.} with product-expansion at least $\rho>0$. Then there exists some $\rho'=\rho'(\rho,t)>0$ depending only on $\rho$ and $t$ such that for every subcode $C_t'\subseteq C_t$ the collection $(C_1,\dots,C_{t-1},C_t')$ has product-expansion at least $\rho'$.
\end{lemma}

To prove Lemma~\ref{lem:pesubcode}, we will use the following definition.

\begin{definition}
  The \textbf{collective product-expansion $P$} of a set $(C_i\subseteq\bF_q^{n_i})_{i\in[t]}$ of classical codes is the largest real number $P\geq 0$ such that for every $m\in\bN$ and every sequence $(c_a\in C_1\boxplus\cdots\boxplus C_t)_{a\in[m]}$ of $m$ dual tensor codewords, there exist decompositions
  \begin{equation*}
    c_a = c_{a,1}+\cdots+c_{a,t}
  \end{equation*}
  with each $c_{a,i}\in C^{(i)}$ such that
  \begin{equation*}
    \left|\bigcup_{a\in[m]}c_a\right| \geq P\sum_{i\in[t]}n_i\left|\bigcup_{a\in[m]}c_{a,i}\right|_i,
  \end{equation*}
  where $|\bigcup_ac_a|$ denotes the number of components on which $c_a$ does not vanish for some $a$ (i.e.~the size of the union of the supports of the $c_a$'s), and $|\bigcup_ac_{a,i}|_i$ denotes the number of direction-$i$ vectors on which $c_{a,i}$ does not vanish for some $a$.
\end{definition}

Lemma~4.13 and Proposition~4.14 in \cite{dinur_expansion_2024} along with Lemma~\ref{lem:pefewercodes} above imply the following:

\begin{lemma}[\cite{dinur_expansion_2024}]
  \label{lem:petocpe}
  Let $\bF_q$ be a field of characteristic $2$, and let $(C_i\subsetneq\bF_q^n)_{i\in[t]}$ be a set of codes with product-expansion at least $\rho>0$. Then there exists some $P=P(\rho,t)>0$ depending only on $\rho$ and $t$ such that for every subset $S\subseteq[t]$, the codes $(C_i)_{i\in S}$ have collective product-expansion at least $P$.
\end{lemma}

We now apply Lemma~\ref{lem:pefewercodes} and Lemma~\ref{lem:petocpe} above to prove Lemma~\ref{lem:pesubcode}.

\begin{proof}[Proof of Lemma~\ref{lem:pesubcode}]
  Fix any $c\in C_1\boxplus\cdots\boxplus C_{t-1}\boxplus C_t'$. Let $n_1=\cdots=n_t=n$. Our goal is to find an appropriate decomposition $c=c_1'+\cdots+c_t'$ with $c_t'\in {C_t'}^{(t;t)}$ and $c_i'\in C_i^{(i;t)}$ for $i\in[t-1]$ such that $|c|\geq\rho'\sum_{i\in[t]}n_i|c_i'|_i$, for an appropriate $\rho'>0$.

  By the product-expansion of $C_1,\dots,C_t$, there exists a decomposition $c=c_1+\cdots+c_t$ each each $c_i\in C_t^{(i;t)}$ such that
  \begin{equation}
    \label{eq:pesubstart}
    \sum_{i\in[t]}n_i|c_i|_i\leq|c|/\rho.
  \end{equation}

  Because $c\in C_1\boxplus\cdots\boxplus C_{t-1}\boxplus C_t'$, Lemma~\ref{lem:homvanexp} implies that
  \begin{equation}
    \label{eq:ctdecomp}
    c_t \in {C_t'}^{(t;t)}+\sum_{i\in[t-1]}(C_i^{(i;t)}\cap C_t^{(t;t)}) = {C_t'}^{(t;t)}+(C_1\boxplus\cdots\boxplus C_{t-1})\otimes C_t.
  \end{equation}
  Let $\Delta=\dim C_t-\dim C_t'$, and let $v_1,\dots,v_\Delta\in C_t$ form a basis for $C_t/C_t'$, meaning that $C_t=\spn\{C_t',v_1,\dots,v_\Delta\}$. Then by~(\ref{eq:ctdecomp}), we can decompose $c_t$ into
  \begin{equation*}
    c_t = c_t'+b_1\otimes v_1+\cdots+b_\Delta\otimes v_\Delta,
  \end{equation*}
  where $c_t'\in{C_t'}^{(t;t)}$, and $b_j\in C_1\boxplus\cdots\boxplus C_{t-1}$ for $j\in[\Delta]$. Furthermore, by definition
  \begin{equation}
    \label{eq:ctbounds}
    |c_t|_t \geq |c_t'|_t \hspace{1em}\text{ and }\hspace{1em} |c_t|_t \geq \left|\bigcup_{j\in[\Delta]}b_j\right|,
  \end{equation}
  as a given direction-$t$ column of $c_t$ must be nonzero if either the associated direction-$t$ column of $c_t'$ is nonzero, or if the associated component of any $b_j$ for $j\in[\Delta]$ is nonzero.

  Because $(C_1,\dots,C_t)$ by assumption have product-expansion at least $\rho>0$, Lemma~\ref{lem:petocpe} implies that that $(C_1,\dots,C_{t-1})$ have collective product-expansion at least $P=P(\rho,t)>0$. Therefore there exist $b_{j,i}\in C^{(i;t-1)}$ for $j\in[\Delta],i\in[t-1]$ such that each $b_j=b_{j,1}+\cdots+b_{j,t-1}$, and
  \begin{equation}
    \label{eq:pesubapplycol}
    \left|\bigcup_{j\in[\Delta]}b_j\right| \geq P\sum_{i\in[t-1]}n_i\left|\bigcup_{j\in[\Delta]}b_{j,i}\right|_i.
  \end{equation}
  Now for $i\in[t-1]$, we define $c_i'\in C^{(i;t)}$ by
  \begin{equation*}
    c_i' := c_i+\sum_{j\in[\Delta]}b_{j,i}\otimes v_j.
  \end{equation*}
  Then by definition
  \begin{equation*}
    \sum_{i\in[t]}c_i' = \left(\sum_{i\in[t-1]}c_i\right)+\left(c_t'+\sum_{j\in[\Delta]}b_j\otimes v_j\right) = \sum_{i\in[t]}c_i = c.
  \end{equation*}
  Furthermore,
  \begin{align*}
    \sum_{i\in[t]}n_i|c_i'|_i
    &\leq \sum_{i\in[t-1]}n_i\left(|c_i|_i+n_t\left|\bigcup_{j\in[\Delta]}b_{j,i}\right|_i\right)+n_t|c_t'|_t \\
    &\leq \sum_{i\in[t]}n_i|c_i|_i + n_t\sum_{i\in[t-1]}n_i\left|\bigcup_{j\in[\Delta]}b_{j,i}\right|_i \\
    &\leq \frac{|c|}{\rho}+\frac{n_t}{P}\left|\bigcup_{j\in[\Delta]}b_j\right| \\
    &\leq \frac{|c|}{\rho}+\frac{n_t}{P}|c_t|_t \\
    &\leq \frac{|c|}{\rho}+\frac{|c|}{P\rho},
  \end{align*}
  where the first inequality above holds by the definition of $c_1',\dots,c_t'$, the second inequality holds by~(\ref{eq:ctbounds}), the third inequality holds by~(\ref{eq:pesubstart}) and~(\ref{eq:pesubapplycol}), the fourth inequality holds by~(\ref{eq:ctbounds}), and the fifth inequality holds by~(\ref{eq:pesubstart}).

  Thus letting
  \begin{equation*}
    \rho'(\rho,t) := \frac{\rho}{1+1/P(\rho,t)},
  \end{equation*}
  we have shown that $(C_1,\dots,C_{t-1},C_t')$ has product-expansion at least $\rho'$, as desired.
\end{proof}

Lemma~\ref{lem:pesubcode} immediately implies Lemma~\ref{lem:pesubmain}:

\begin{proof}[Proof of Lemma~\ref{lem:pesubmain}]
  The desired result follows by applying Lemma~\ref{lem:pesubcode} $t$ times, once for each $i\in[t]$ to replace $C_i$ with $C_i'$ in the product-expanding collection of codes.
\end{proof}

\end{document}